\documentclass[nonacm,sigplan]{acmart}

\usepackage{multirow}
\usepackage{qcircuit}
\usepackage{braket}

\newcommand{\gln}[1]{\mathbb{GL}_{#1}(\mathbb{F}_2)}
\newcommand{\kpsi}[1]{\ket{\psi_{#1}}}

\newcommand{\st}{\ s.t.\ }

\newcommand{\aut}[1]{\operatorname{Aut}\!\left(#1\right)}
\newcommand{\orbg}[1]{\operatorname{Orb}^L_{\mathcal{G}}\!\left(#1\right)}
\newcommand{\orbs}[1]{\operatorname{Orb}^R_{\Sigma}\!\left(#1\right)}

\newcommand{\wt}[1]{\operatorname{wt}\!\left(#1\right)}
\newcommand{\rows}[1]{\operatorname{Rows}\!\left(#1\right)}
\newcommand{\cols}[1]{\operatorname{Cols}\!\left(#1\right)}

\newcommand{\rs}[1]{\left\langle #1 \right\rangle}
\newcommand{\ns}[1]{\left\langle #1 \right\rangle^{\!\perp}}

\newtheorem{theorem}{Theorem}
\newtheorem{proposition}{Proposition}
\newtheorem{corollary}[proposition]{Corollary}

\numberwithin{proposition}{section}  

\newtheorem{remark}{Remark}
\newtheorem{fact}[remark]{Fact}

\newtheorem{definition}{Definition}

\begin{document}

\title{Spacetime-Efficient and Hardware-Compatible Complex Quantum Logic Units in qLDPC Codes}

\author{Willers Yang$^*$, Jason Chadwick, Mariesa H. Teo, Joshua Viszlai, Frederic T. Chong\\
\textit{Department of Computer Science, University of Chicago} \\Chicago, IL, USA}
\thanks{\href{mailto:willers@uchicago.edu}{$^*$willers@uchicago.edu}}

\begin{abstract}
Quantum low-density parity-check (qLDPC) codes offer a promising route to scalable fault-tolerant quantum computing (FTQC) due to their substantially reduced footprint. However, these gains can be diluted at utility scale if we cannot also realize space-time-efficient logical operations for relevant quantum applications. We present \emph{RASCqL}, a \underline{\textit{R}}eaction-time-limited \underline{\textit{A}}rchitecture for \underline{\textit{S}}pace-time-efficient \underline{\textit{C}}omplex-Instruction-Set Quantum Computation with \underline{\textit{q}}LDPC \underline{\textit{L}}ogic. RASCqL supports key algorithmic subroutines--such as quantum arithmetic and state preparation--directly within co-designed qLDPC codes, achieving $2\times$–$7\times$ reductions in qubit footprint while maintaining space-time volume comparable to state-of-the-art transversal surface-code architectures.

Unlike prior approaches aiming for versatile logical instruction sets amenable to arbitrary circuits, RASCqL adopts an application-tailored code modification that embeds specific complex Clifford transformations useful for common subroutines as virtually implementable operations from code automorphisms. RASCqL further leverages parallel physical operations available in reconfigurable neutral-atom arrays to enable fast QEC cycles and high-fidelity transversal operations. At the cost of increased design complexity and specialization, RASCqL has the potential to improve end-to-end resource estimates for important applications such as factoring and quantum chemistry simulation in both footprint and space-time volume under realistic physical error rates of $2\times10^{-3}$–$5\times10^{-4}$, without requiring additional hardware capabilities. These results demonstrate that qLDPC codes can serve as \emph{Complex Quantum Logic Units} for useful quantum algorithms, extending their practical utility in fault-tolerant quantum computing architectures.
\end{abstract}

%%%%%%%%%%%%%%%%%%%%%%%%%%%%%%%%%%%%

\maketitle
\pagestyle{plain}

%%%%%% -- PAPER CONTENT STARTS-- %%%%%%%%

\begin{figure}
    \centering
    \includegraphics[width=0.95\linewidth]{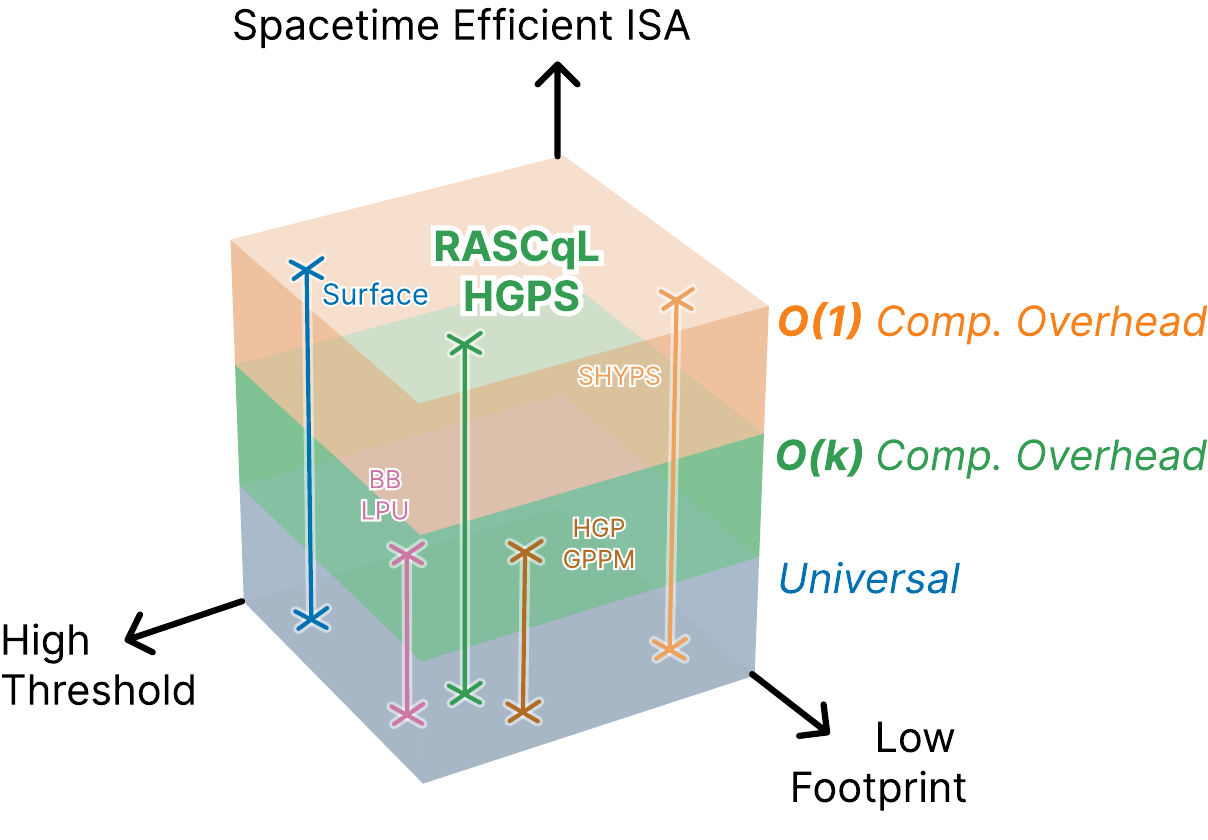}
    \includegraphics[width=0.95\linewidth]{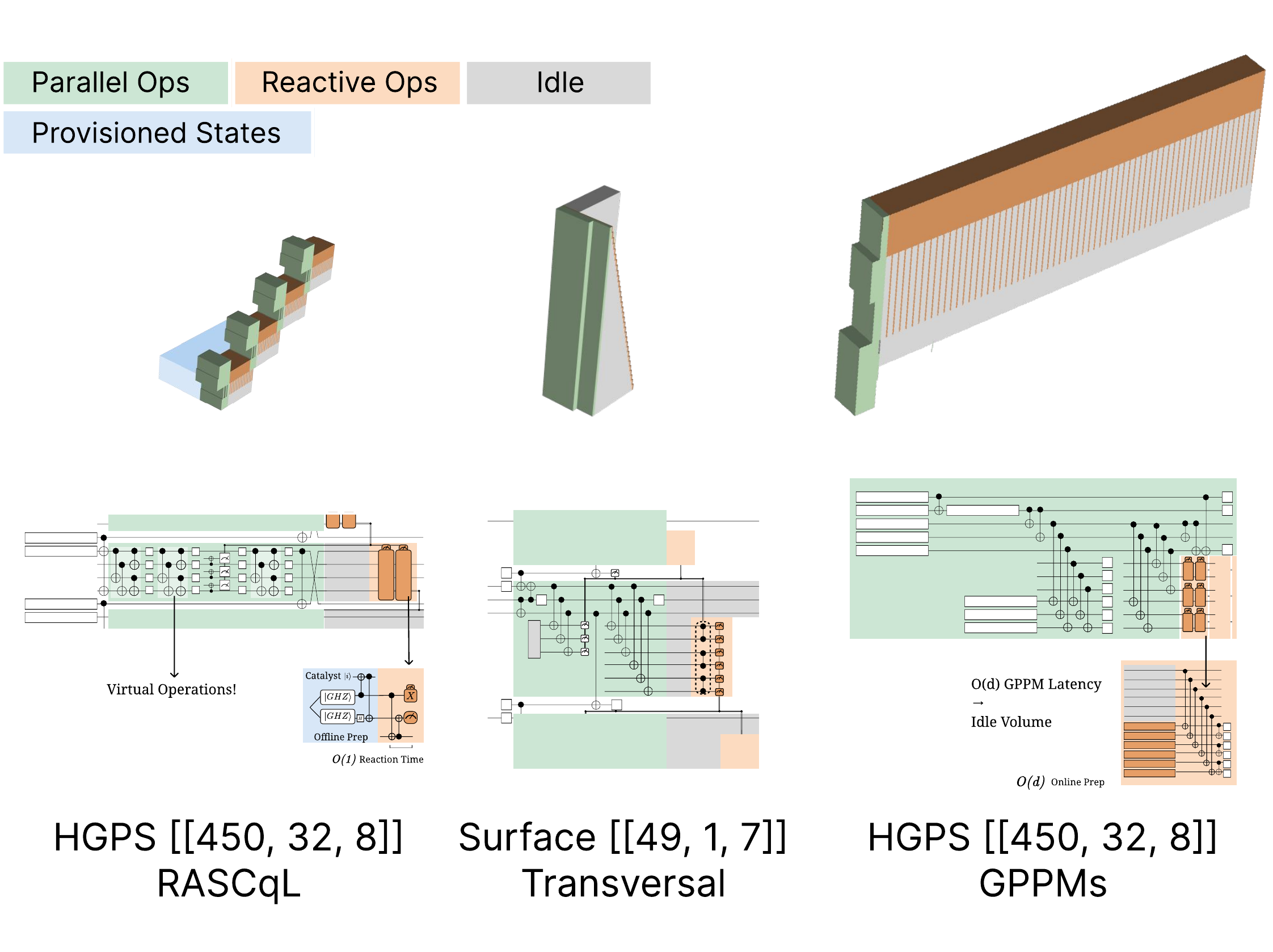}
    \caption{Top: RASCqL exhibits high circuit-level threshold, high encoding rate, and space-time efficient ISA for algorithms. Bottom: space-time volume comparison of RASCqL adder against baselines.}
    \label{fig:comp}
\end{figure}
\section{Introduction}
Quantum computers promise polynomial complexity algorithms for many tasks currently without efficient classical solutions, including factoring \cite{Shor1994} and chemistry simulations \cite{Feynman1982}. Implementing useful quantum algorithms requires a fault-tolerant quantum computer (FTQC) with logical qubits and operations achieving error rates as low as $10^{-9}$ to $10^{-16}$ \cite{Accessing}, while current devices operate with physical error rates around $10^{-3}$ \cite{GoogleSC, QuantinuumH2, Bluvstein2024}. Quantum error-correcting codes (QECC) are deployed to bridge the gap between useful quantum advantage and noisy hardware, providing exponential error suppression when physical operations meet a code-specific error threshold and defining a logical instruction-set architecture capable of universal quantum computation.

The \textit{surface code} is a popular QECC due to its simple hardware requirements and expressive logical primitives, but suffers from overheads that scale polynomially in the number of correctable errors. Its syndrome extraction circuits require only 2D nearest-neighbor interactions, matching to planar superconducting devices with tunable couplers, and it tolerates physical error rates of around $1\%$ \cite{Fowler2012surface}, allowing experimental demonstrations of error suppression in existing devices \cite{GoogleSC}. Logical operations on surface codes form a compact, uniform instruction set, enabled by hardware-compatible lattice-surgery primitives for Pauli-based computation \cite{lattice_surgery}, or constant-time transversal gates on reconfigurable neutral-atom arrays (RNAA) with correlated decoding \cite{algoFT,Zhou_2025}. These FT primitives, along with resource states prepared through distillation protocols \cite{hastingshaah,cultivation}, can be composed efficiently to execute arbitrary quantum programs. Consequently, optimizing a surface code architecture often parallels RISC design: a standard, expressive ISA, heavily optimized across layers from algorithm mapping and resource estimation \cite{GidneyRSA2019, LitinskiElliptic2023, Accessing, GidneyLinearT, LitinskiActiveVolume} to logical-level compilation \cite{edp_beverland, Litinski2019gameofsurfacecodes} and hardware co-designs \cite{bias-tailored, erasures}. The cost of simplicity, however, is that the surface code (as is the case with any codes with 2D-local checks) suffers from a qubit overhead that scales quadratically with the number of correctable errors, leading to resource estimations at $1000\times$ physical qubits to encode a single logical qubit \cite{gidney2025factor2048bitrsa}.

\emph{Quantum low-density parity-check (qLDPC) codes}, on the other hand, do not necessarily suffer from space-time overheads polynomial in $d$, at the cost of increased hardware complexity and limited instruction sets. Recently, studies on the algebraic structures of qLDPC codes have led to many hardware-efficient co-designs matching restricted code families to various physical architectures \cite{IBMBB, HGPconstantAtoms, JoshBB}, turning them into promising alternatives to surface codes. When deployed as high-threshold, low-footprint quantum memories, such codes can expand the frontier of achievable quantum algorithms on qubit-constrained devices \cite{stein2024architecturesheterogeneousquantumerror}, potentially shortening the timeline for demonstrations of useful quantum advantage. Another disadvantage qLDPC codes face is access to efficient logical operations. At the asymptotic limit, there exists a fundamental trade-off between a qLDPC code's memory performance and the efficiency of existing fault tolerant instructions \cite{guyot2025addressabilityproblemcsscodes}, and in practice, ISAs of existing qLDPC codes are either limited in size and expressivity \cite{IBMBB, danbrownauto}, or require physical device capabilities beyond near-term devices \cite{ComputingQLDPC}.

While it appears challenging to realize utility scale general-purpose computation in qLDPC codes that simultaneously achieve high threshold, high rate, and space-time efficient ISA, these limits only hold if each block must independently support arbitrary circuits efficiently. Instead, we propose code constructions that possess limited but useful automorphism gates that, while not amenable to generic inputs, can directly accelerate common subroutines such as quantum adders, magic-state distillation, and quantum lookup. Unlike other proposals for logic in qLDPC codes that aim at a RISC ISA capable of handling arbitrary circuits, our perspective reframes in-block computation on qLDPC codes as a \emph{Complex-Instruction-Set Quantum} (CISQ) architecture, where each code is co-designed to implement a limited, functionally-targeted set of logical operations natively and efficiently. For a visual example, Figure~\ref{fig:comp} shows how the quantum adders are compiled in RASCqL in comparison with other approaches. 

A specialized logical ISA is particularly relevant for existing fault tolerant quantum workloads, where the space of useful applications is narrow with implementations dominated by few functional building blocks \cite{Zhou_2025}. In particular, not only do they support factoring \cite{gidney2025factor2048bitrsa}, adders also appear as the dominant cost in various chemistry applications \cite{Low_2019, harrigan2024expressinganalyzingquantumalgorithms,von_Burg_2021, Lee_2021, Gily_n_2019, caesura2025fasterquantumchemistrysimulations}.
As we will show, designing codes that excel at these common subroutines, rather than universally supporting arbitrary operations, is the key to efficient in-block computation with qLDPC codes.

We present \emph{RASCqL}, a \underline{R}eaction-limited \underline{A}rchitecture for \underline{S}pace-time efficient \underline{C}omplex \underline{Q}uantum \underline{L}ogic, where we enable efficient and hardware compatible in-block computation in co-designed qLDPC codes with the potential to outperform surface codes. RASCqL makes three key contributions:
\begin{enumerate}
    \item \textbf{\emph{Complex Quantum Logical Units (CQLU)}} We design complex quantum logical units on co-designed qLDPC codes with useful  complex quantum instructions and develop targeted SIMD compilations of key subroutines onto CQLU instructions. For adders, we achieve up to 7x footprint reduction with 1.25x Clifford-volume reduction against state-of-the-art transversal surface code baselines in a reaction-time limited compilation framework.
    \item \textbf{\emph{Predictive Resource-state Preparation (PReP)}} We predict resource state requirements for fast reactive operations and pipelined CQLU production. In particular, we achieve reactive Pauli Product measurements, arbitrary pattern CNOT fan-outs, and non-Clifford gates with $O(1)$ latency in expectation. Our optimized state factories achieve $10\times$ volume reduction for $\ket{GHZ}$ states compared to surface code baselines, and demonstrate high fidelity magic states in LDPC codes with $2\times$ footprint reduction and $3\times$ increase in volume.
    \item \textbf{\emph{Physical implementations on reconfigurable neutral atom arrays}} We give explicit layout and AOD-compatible movement schedules to all primitives, achieving QEC cycles in millisecond scale, demonstrating subthreshold error suppression while carefully accounting for additional idle and gate error due to transversal logic, and produce realistic resource estimation for key subroutines used in factoring and chemistry simulations. 
\end{enumerate}

Together, these designs establish a CISQ framework and expand the practical utility of qLDPC codes beyond low-footprint memory. To our knowledge, RASCqL is the first to propose a class of hardware-compatible and high rate qLDPC codes that demonstrate quantum arithmetic modules with improved space-time volumes compared to state-of-the-art surface code baselines. %and in-block MSD with 

The rest of the paper is organized as follows. We begin with an introduction to FTQC as a three layer stack in Section~\ref{sec:background}. Then, we describe RASCqL in Section~\ref{sec:arch}. Code constructions, targeted compilations, resource state prediction, and physical RNAA implementations are described in Section~\ref{sec:methods}, where detailed constructions and proofs are deferred to the Appendix. Finally, we compare the performance of RASCqL on subroutines including adders, GHZ and magic state preparation to state-of-the-art surface code baseline in Section~\ref{sec:eval}, where explicit physical movement and gate schedules are used for resource estimation and a detailed circuit-level noise model is used for logical error rate simulations.

\begin{figure}
    \centering
    \includegraphics[width=0.85\linewidth]{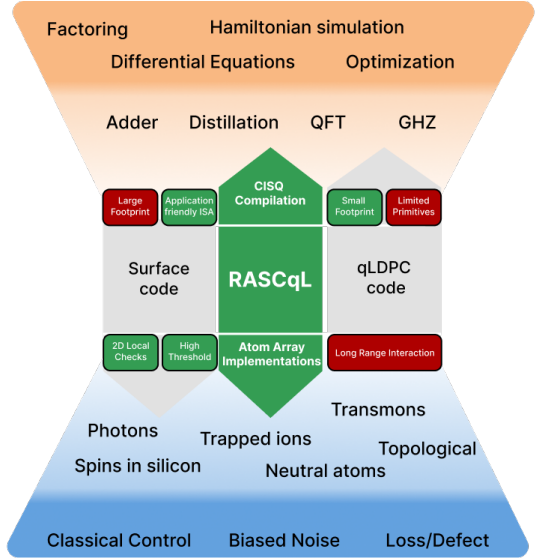}
    \caption{A three-stack view of FTQC. The surface code implements an efficient logical layer due to its hardware compatibility and RISC ISA, but suffers from large qubit overhead. qLDPC codes may reduce footprint, but require long-range interactions and have comparatively limited logical primitives. RASCqL achieves efficient in-block computation by co-designing qLDPC codes with CISQ compilations with key functional subroutines and reconfigurable neutral atom-array implementations. }
    \label{fig:stack}
\end{figure}
\section{Background}\label{sec:background}

A fault-tolerant quantum computer consists of a three-layer stack (Figure~\ref{fig:stack}): the \emph{functional layer}, where algorithms and subroutines are described by abstract quantum programs; the \emph{logical layer}, where QECCs provide error suppression and fault-tolerant primitives that define the instruction set architecture (ISA); and the \emph{physical layer}, where hardware parameters such as connectivity and noise model determine implementability, costs, and fidelity of QEC primitives. The logical layer forms a thin waist between applications and diverse hardware, introducing a vast design space rich in optimization opportunities and potential inefficiencies.

\subsection{Logical Layer: Quantum Error Correcting Codes}

A quantum error correcting code (QECC)'s error suppression performance is summarized by four key parameters--$n,k,d$ and $p_{th}$. A $[[n,k,d]]$ code encodes $k$ logical qubits using $n$ physical qubits, where a Pauli operator of minimum weight $d$ is required to transform between logical states. $p_{th}$ refers to the \emph{error threshold}, or the physical error rate below which a code begins to work reliably. A QECC can correct up to $\lfloor\frac{d}{2}\rfloor$ errors below threshold by performing \textit{syndrome extractions}, outcomes of which are used by a \emph{decoder} to identify most-likely errors and compute corrections. In an error model where each physical Pauli error occurs independently with probability $p_{phys}$ and assuming perfect syndrome extraction, the logical error rate (LER) can be exponentially suppressed according to the formula:
\begin{equation}\label{eq:qec}
    p_{L} \approx (\frac{p_{phys}}{p_{th}})^{\lfloor\frac{d}{2}\rfloor+1}. 
\end{equation}
While $d$ and Equation~\ref{eq:qec} describe the theoretical capacity of a QECC, errors in realistic systems often deviate from idealistic assumptions. We can more accurately assess the LER performance of a code using circuit-level simulations that account for error propagation through syndrome extraction circuits and device-specific noise models.

The \emph{encoding rate}, $\frac{k}{n}$, captures the code's space efficiency; codes with lower rate will have larger per-logical qubit footprint on a device. Due to the parallelizability afforded by backwards-in-time gate teleportation \cite{LitinskiActiveVolume}, the space-time circuit volume is an important metric used to resource estimate quantum algorithms and evaluate FTQC system performance, in addition to LER. Current estimates place the circuit volume of realistic implementations of practical quantum algorithms such as factoring and chemistry simulation at millions of qubit-days \cite{Accessing}, in part due to the large qubit-overhead associated with the leading proposal of using surface codes as the logic layer.

On of the costs of error suppression is the increased difficulty in implementing an efficient, universal logical ISA. Indeed, a distance-$d$ error protection also places the same restriction on the number of physical operations needed to effect logical transformation. In order to maintain fault tolerance, these logical operations must be implemented carefully to avoid propagating physical errors within the support of a logical qubit and reducing distance, while guaranteeing that additional errors from these physical operations are also detectable. Logical operations satisfying these conditions are said to be \emph{transversal}. The existence and expressivity of transversal logical gates depend on inherent symmetries of the code in question, and no QECC can implement a universal set of logical operations transversally \cite{Knill}. Consequently, an efficient logical ISA requires two components: an efficient set of native fault-tolerant gadgets, usually a generating set of Clifford gates such as the CNOT, H, and S gates, or Pauli Product measurements, with a purification protocol to prepare non-native magic states to achieve universal computation \cite{Bravyi_2005}. In Figure~\ref{fig:isa-comp}, we summarize the key techniques for enabling fault-tolerant logical operations. The availability and costs of these gates further depend on the code and technique for logical operations. 
 \begin{figure*}[ht!]
     \centering
     \includegraphics[width=\linewidth]{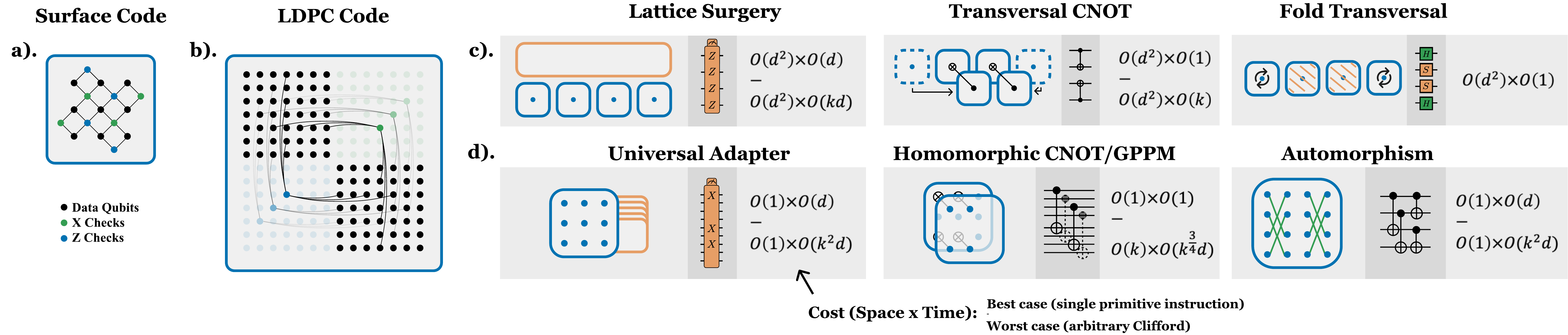}
     \caption{Quantum error correcting codes. a) $[[9,1,3]]$ surface codes with 2D planar syndrome checks. b) $[[98, 18, 4]]$ qLDPC code with non-local syndrome checks and improved encoding rate. c) Fault tolerant logical instructions of a surface code: (from left to right) arbitrary Pauli Product measurements, parallel transversal CNOTs, and addressable single qubit Clifford gates. d) Fault tolerant logical instructions of qLDPC codes: particular Pauli Product measurements given by a LPU, global transversal CNOTs, and automorphism Clifford gates given by code symmetries.}
     \label{fig:isa-comp}
 \end{figure*}
\subsubsection{The surface code}
The \emph{surface code}\footnote{We focus on the rotated surface code~\cite{bombin2007optimal}, which achieves a higher encoding rate and threshold among variants.} is a $[[d^2,1,d]]$ QECC with a threshold of about $1\%$ under a fast matching decoder~\cite{Fowler2012surface}. Analogous to the classical repetition code, it has a particularly simple construction--on a $D\times D$ grid of physical qubits, a surface code interleaves an equal number of $X$ and $Z$ checks that each covers four 2D-nearest-neighbor data qubits in a symmetric pattern, as illustrated in Figure~\ref{fig:isa-comp}. As is the case for all QECCs with local checks, the surface code suffers from a qubit overhead that scales quadratically with distance. 

There are two approaches for logical operations on surface codes. On fixed-topology hardware such as superconducting qubits, we can perform arbitrary multi-qubit Pauli-product measurements (PPM) using lattice surgery with an additional $O(d)$ time overhead, leading to an overall space-time volume of $O(d^3)$ qubit-cycles for each instruction\cite{Litinski2019gameofsurfacecodes}. In addition, CNOT, S, and H gates can be performed transversally on RNAA with $O(1)$ time overhead using correlated decoding \cite{algoFT}, reducing the space-time overhead to $O(d^2)$ per instruction.

\subsubsection{qLDPC codes}
A QECC is said to be low-density parity-check if the number of other physical qubits each qubit interacts with is upper bounded during an error correction cycle. This LDPC condition limits how far a physical error can propagate, and is necessary to ensure fault-tolerance. While the surface code is also a qLDPC code, general constructions do not impose locality constraints that lead to undesirable overheads. The first constant rate qLDPC code discovered uses good classical LDPC codes as seeds in a Hypergraph-Product (HGP) construction \cite{Leverrier_2015}. More recently, high rate qLDPC codes with concrete physical implementations have also been discovered that exhibit exponential error suppression under realistic physical error rates \cite{IBMBB,HGPconstantAtoms}.

We describe three main approaches to enable FT logical operations in qLDPC codes. While potentially more space-time efficient, they are often either challenging to compile to or implement on hardware. First, we can realize a lattice-surgery style RISC as ancilla assisted Pauli-product measurements, which are also referred to as Logical Processing Units (LPU)\cite{cross2025improvedqldpcsurgerylogical}. Leading proposals prioritizes hardware compatibility with small footprint LPUs and favorable planar embedding \cite{yoder2025tourgrossmodularquantum}, while having limited native instructions. While the space-time volume of native LPU instructions can out-perform surface code lattice surgery, non-native operations need to be decomposed with $O(k)$ time overhead on average, which may eclipse the space savings in practice. Secondly, on physical platforms that support qubit rearrangement, Clifford operations and PPMs can be performed in parallel using transversal CNOTs or fold-transversal CZ-S gates \cite{partition,hgphomo, bacthed}. Such operations can be more space-time efficient than surface code operations when performed in batch, but only when the logical circuits we wish to implement follow the same fold- or grid-like structures. Finally, a less versatile proposal uses code automorphism--intrinsic symmetries of the QECC--that in rare cases allow complex logical operations to be implemented transversally or as virtual qubit relabeling. However, stringent conditions need to be met to lift these symmetries to fault tolerant logical operations \cite{berthusen2025automorphismgadgetshomologicalproduct}. On qLDPC codes with more than a couple qubits, automorphism-based operations often manifest as arbitrary global transformations that do not find direct applications easily \cite{danbrownauto}, leading them to serve mostly as memory operations such as qubit permutations \cite{stein2024architecturesheterogeneousquantumerror, yoder2025tourgrossmodularquantum, hgphomo}.
%We can perform logical operations on qLDPC analogously to surface code techniques. Universal adapters generalizes lattice surgery to perform PPMs \tocite{UniversalAdapters line of work}. Similar to lattice surgery, an ancillae system is initialized to extract and measure specific logical degrees of freedoms; however, while surface code lattice surgery can implement any arbitrary PPM with similar costs, the number and parallelism of implementable PPMs in LDPC codes is limited by the size of adapters used. While a small number of native PPMs can be executed in a single logical cycle, arbitrary measurements in general are composed using $O(k)$ such native PPMs, making the worst-case compilation overhead worse than surface codes. On RNAA, [describe]

%The HGP code inherits several good properties from its parent codes--First, it's LDPC if both classical codes are LPDC; in particular, if the parity checks of the parent codes have weight at most $w_i$, $\mathcal{Q}_{HGP}$ has stabilizers with weight at most $w_1+w_2$. Secondly, the quantum code has encoding rate $\Theta(\frac {k^2}{n^2}) = \Theta(1)$ if each $\Theta(\frac {k}{n}) = \Theta(1)$. At practically relevant distances ($d\in [8,30]$), there exists small instances of good qLDPC codes that require up to $10\times$ less qubits than surface codes, with proposed hardware implementations on thin-planar superconducting chips \tocite{BB} and RNAA \tocite{constant overhead}.

\subsubsection{Magic State Distillation}
A key limitation to QECC is the aforementioned lack of universal native instructions. A standard approach is to supplement the logical ISA with high fidelity resource states that can implement non-native operations fault tolerantly. These states can be prepared using a magic state distillation (MSD) protocol, where fault-tolerant Clifford operations are used to purify noisy magic states. For example, the $[[15,1,3]]$ protocol outputs a $\ket{T}$ state at error rate $O(p^3)$ with high probability, using 15 noisy $\ket{T}$ states at physical error rate $O(p)$. These magic states can then be used to synthesize universal gates with $O(\log\frac{1}{\epsilon})$ overhead to a target error rate $\epsilon$ \cite{gridsynth}. Due to its high overhead, it is estimated that as high as $95\%$ of computing volume may be spent on producing high-fidelity magic states \cite{fowler}, making it an important bottleneck to FTQC. In this work, we will focus on a family of MSDs described by Hastings and Haah \cite{hastingshaah} which includes the $[[15,1,3]]$ codes, though our methodology extends to Reed-Muller codes and triorthogonal codes more generally. 

\subsection{Physical Layer: Reconfigurable Atom Arrays}
qLDPC codes and transversal logic require high-fidelity non-local interactions. One approach is to use platforms that support programmable qubit movement, implementing flexible connectivity through physical reconfiguration, such as RNAAs where neutral atom qubits can be moved in parallel using optical tweezers
%enabling  experimental demonstrations of a qLDPC code on a 280-qubit atom array 
\cite{Bluvstein2024}. Several other proposals exist to realize these requirements in hardware, including modular architectures with photonic interconnects \cite{LitinskiActiveVolume,XanaduPhotonic}, all-to-all connected trapped-ion qubits with long-range Coulomb-mediated gates \cite{ionQforte}, and superconducting architectures with a few interconnected layers of planar tunable couplers \cite{IBMBB}. While these approaches provide pathways to long-range connectivity, current experimental demonstrations are limited in scale—typically in the range of a few to a few dozen qubits \cite{QuantinuumH2, ionQforte, XanaduPhotonic, IBMroadmap}—well below requirements for utility-scale FTQC. On the other hand, up to 6100 qubits have been demonstrated on neutral atom processors \cite{Manetsch_2025}. For this reason, we will focus on RNAA in our designs and analysis.

Neutral atom arrays have been demonstrated to have seconds-long coherence times, with high fidelity single- and two-qubit gates \cite{Bluvstein2024, Manetsch_2025, bedalov2024faulttolerantoperationmaterialsscience}. An RNAA device is bottlenecked by movement and measurements--while physical operations only take on the order of micro-seconds, atom motion and measurement can take orders of magnitude longer, up to half a millisecond \cite{Zhou_2025}. One prevalent architecture model for neutral atoms is the zone-based architecture, which separates space into storage zones, entangling zones, and readout zones, into which qubits have to be moved to undergo zone-specific operations \cite{Bluvstein2024}. This architecture requires motion before every entangling gate and measurement, which can become a limiting factor in computation. Instead, we focus here on an architecture combining motion with free-space measurement, as is done in the baseline work we compare to \cite{Zhou_2025}, and use the same set of hardware parameters for a fair comparison. The logical operation of such an architecture has been demonstrated at a small-scale, with parallelized qubit control and acceleration-based movement (the duration of motion scales as a square root of the distance moved) \cite{bedalov2024faulttolerantoperationmaterialsscience, rines2025demonstrationlogicalarchitectureuniting}.

\begin{figure}
    \centering
    \includegraphics[width=\linewidth]{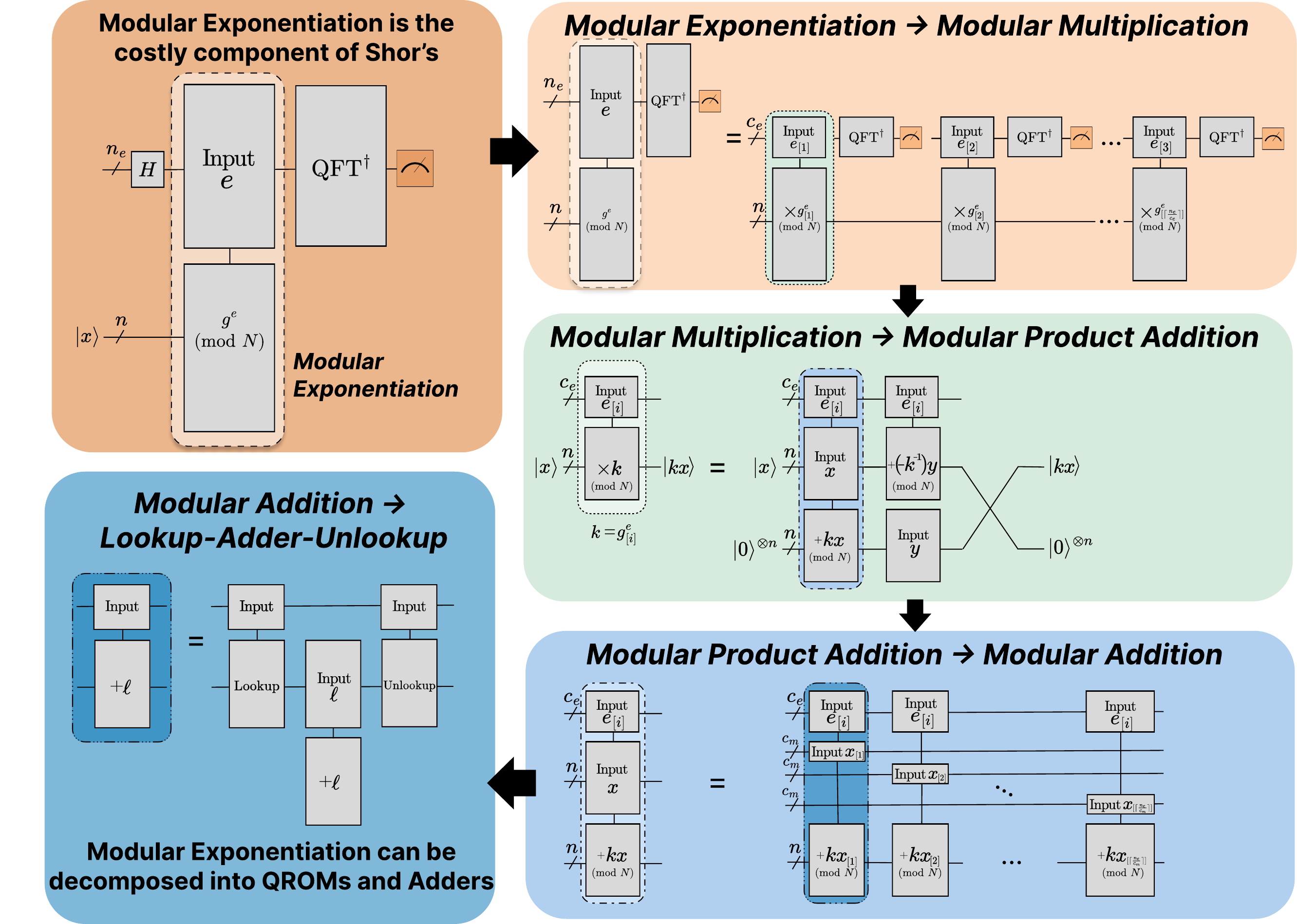}
    \caption{Abstract overview of the factoring algorithm. It can be implemented using repeated applications of quantum adders and quantum look-up tables, with a QFT.}
    \label{fig:factor}
\end{figure}
\subsection{Functional Layer: Quantum Algorithms}
Developing quantum algorithms with a speedup over classical methods has been an active area of research since the emergence of quantum computing, resulting in algorithms for integer factoring \cite{Shor1994}, search and optimization \cite{Grover1996}, combinatorial optimization \cite{Farhi2014QAOA,comb_opt_heuristics2020}, quantum chemistry \cite{Low_2019, harrigan2024expressinganalyzingquantumalgorithms,von_Burg_2021, Lee_2021, Gily_n_2019, caesura2025fasterquantumchemistrysimulations, bauer_quantum_2020}, and more \cite{montanaro_quantum_2016, cerezo_variational_2021, HHL2009}. Here, we introduce two prominent examples of FTQC applications, factoring and quantum chemistry, which feature common subroutines our CISQ architecture can be optimized for. 

\subsubsection{Factoring}
Shor's algorithm solves the integer factoring problem -- given $N$, the aim is to find factors $p$ and $q$ such that $N=p \times q$. The algorithm is a hybrid classical-quantum algorithm whose quantum component focuses on period-finding, or in other words, for some $a$, finding the smallest $r$ such $a^r = 1 \text{ mod N}$. Period-finding can be broken down into a modular exponentiation, followed by a quantum Fourier transform. As the modular exponentiation dominates the cost of period-finding, many resource estimates focus on this part of the algorithm \cite{GidneyRSA2019, gidney2025factor2048bitrsa, Zhou_2025}, which can be implemented using windowed arithmetic that decomposes into sequences of adder and quantum read-only memory (QROM) subroutines \cite{gidney2019windowedquantumarithmetic} (see Figure \ref{fig:factor}).
 
\subsubsection{Quantum Chemistry}
These subroutines are also useful in quantum chemistry applications which use qubitization and quantum walk-based methods. Among other chemistry applications that utilize quantum arithmetic  \cite{Low_2019, harrigan2024expressinganalyzingquantumalgorithms,von_Burg_2021, Lee_2021, Gily_n_2019, caesura2025fasterquantumchemistrysimulations, bauer_quantum_2020}, one salient application is determining the ground states and ground state energy of a given Hamiltonian \cite{bauer_quantum_2020, efficientquantumchemistry2021, Low2019hamiltonian}. These algorithms use PREPARE and SELECT oracles to encode a Hamiltonian, where PREPARE blocks perform state preparation of a superposition state and SELECT blocks implement the $\ell$th term in the Hamiltonian on the target register if the ancilla register is in the state  $|\ell\rangle$\cite{GidneyLinearT, efficientquantumchemistry2021, caesura2025fasterquantumchemistrysimulations}. These oracles, in turn, are largely composed of adder and QROM subroutines, making them crucial complex instructions for an efficient FTQC.

The utility of adders and QROM extends beyond factoring and quantum chemistry: the adder subroutine is used more generally to implement quantum arithmetic, $n$ single-qubit rotations, and the $n$-qubit phase gradient operation \cite{gidney_halving_2018}; QROMs are key components of protocols for state preparation and unitary synthesis \cite{Low2024tradingtgatesdirty}. Building on the insight that several state-of-the-art quantum algorithms decompose into a finite set of common functional subroutines, the performance of a logical architecture can largely be captured by how efficiently it can implement a select number of subroutines. The decomposition of large-scale algorithms into core subroutines points to a scalable CISQ strategy that leverages the efficient operations unique to specific qLDPC codes. 

We note that while we focus in this paper on construction and resource estimates for the adder, the methods we describe, such as the implementation of CNOT- fanouts and Toffoli ladders, extend to QROMs as well.

\subsection{Full Stack Fault-Tolerant Quantum Architectures}
Resource costs for FTQC have continuously improved over the past decade due to improving encoding rates and more efficient techniques for logical action. More recently, ISA (co-)design and memory hierarchy proposals have unlocked additional gains from the system level. 

Initial proposals for an FTQC relied on multiple layers of concatenated codes such as Steane codes \cite{steane}. These schemes provided a clear theoretical foundation, but their resource costs scale exponentially to the number of correctable errors. On the other hand, surface code architectures promised only quadratic qubit overhead. Logical computation on surface codes has evolved substantially: early proposals relied on defect braiding \cite{Fowler2012surface}, while later approaches introduced lattice surgery for more space-time-efficient merge-split operations \cite{Litinski2019gameofsurfacecodes}. The state of the art (SoTA) architecture proposal on reconfigurable atom arrays uses transversal Clifford operations  correlated decoding, which we use as a baseline for comparison \cite{Zhou_2025}.

qLDPC codes can potentially achieve 10x qubit-reduction with near-term devices, while at present they often struggle to implement algorithms due to the lack of generically applicable logical ISA. 
One way around this is to use qLDPC as low-footprint quantum memory in a hierarchical system, where surface codes are still used for logical computations \cite{IBMBB,HGPconstantAtoms, stein2024architecturesheterogeneousquantumerror,JoshBB, gidney2025factor2048bitrsa}. 
While achieving a reduction in footprint, such systems may be bottle-necked by I/O for applications with high parallelism\cite{stein2024architecturesheterogeneousquantumerror, JoshBB}, and may be less efficient from the space-time cost perspective. With various methods to trade between space and time costs\cite{fowler2013timeoptimalquantumcomputation}, compilations of practical algorithms that targets space-time volume can avoid inherent program-serializations to more effectively mitigate the impact of high idle volumes \cite{LitinskiActiveVolume}, which may limit the range for memory hierarchy proposals to find substantial application.

While there exist certain efficient logical operations that are more space-time efficient in qLDPC codes \cite{hgphomo, ParallelPauliMeasure}, it remains challenging to show competitive space-time performance at utility scale for generic algorithms, which are often compiled with a surface-code inspired RISC instruction set. Attempts to recreate such RISC ISAs on qLDPC codes struggle with mismatched priorities across the stack, resulting in either compute-efficient constructions that require capabilities beyond intermediate-term devices, or restrict to hardware-compatible designs at the cost of space-time efficiency. For example, the Subsystem Hypergraph Simplex (SHYPS) codes can implement generic Clifford operations with constant space-time overhead, while its stabilizers have weights that scale with code length, requiring orders of magnitude stricter requirements for physical error rates at high distances. On the other hand, IBM's Bivariate Bicycle (BB) codes with logical processing units \cite{yoder2025tourgrossmodularquantum} prioritize hardware compatibility, maintaining both LDPC and bi-planarity properties and achieving high circuit-level threshold, but implements a limited ISA that leads to long gate decomposition sequences when compiling for practical algorithms.

RASCqL takes the natural next step in the evolution of quantum architectures by viewing qLDPC codes as specialized accelerators for complex instructions, and shows that overall space-time reduction is attainable, even with realistic error rates for practical algorithms. We pay the cost in terms of flexibility and increased design complexity, which we argue is a worthwhile trade-off for existing quantum workloads that require few key subroutines. Unlike memory hierarchy proposals, RASCqL implements all logical operations in-block in a reaction-time limited framework, avoiding costly inter-code communication, while retaining footprint advantage. RASCqL targets a practical class of algorithms such as factoring and chemistry simulations, and unlike codes designed to be efficient for generic inputs such as $SHYPS$, RASCqL remains LDPC, with a circuit-level threshold of around $0.78\%$. RASCqL also features instructions that can be implemented as virtual qubit relabeling. The use of virtual logical operations for logical computation is proposed independently in a concurrent work that studies phantom codes--codes that can implement all logical entangling gates as virtual qubit relabeling \cite{koh2026entanglinglogicalqubitsphysical}. Although the code constructions proposed here can be used to turn a qLDPC code into a phantom code, doing so may incur exponential overhead and cannot guarantee LDPC properties in general.  Instead, RASCqL enables a small but sufficient set of entangling gates while optimizing for footprint and logical performance. Finally, RASCqL is hardware-compatible, with structures that match RNAA capabilities. Compared to other hardware efficient approaches, while RASCqL's ISA is not necessarily more expressive, they are co-designed to contain complex instructions that are efficient for applications, accompanied with targeted compilations of key subroutines that are sufficient for the majority of factoring and chemistry simulations.

\section{RASCqL Architecture}\label{sec:arch}
\begin{figure}
    \centering
    \includegraphics[width=\linewidth]{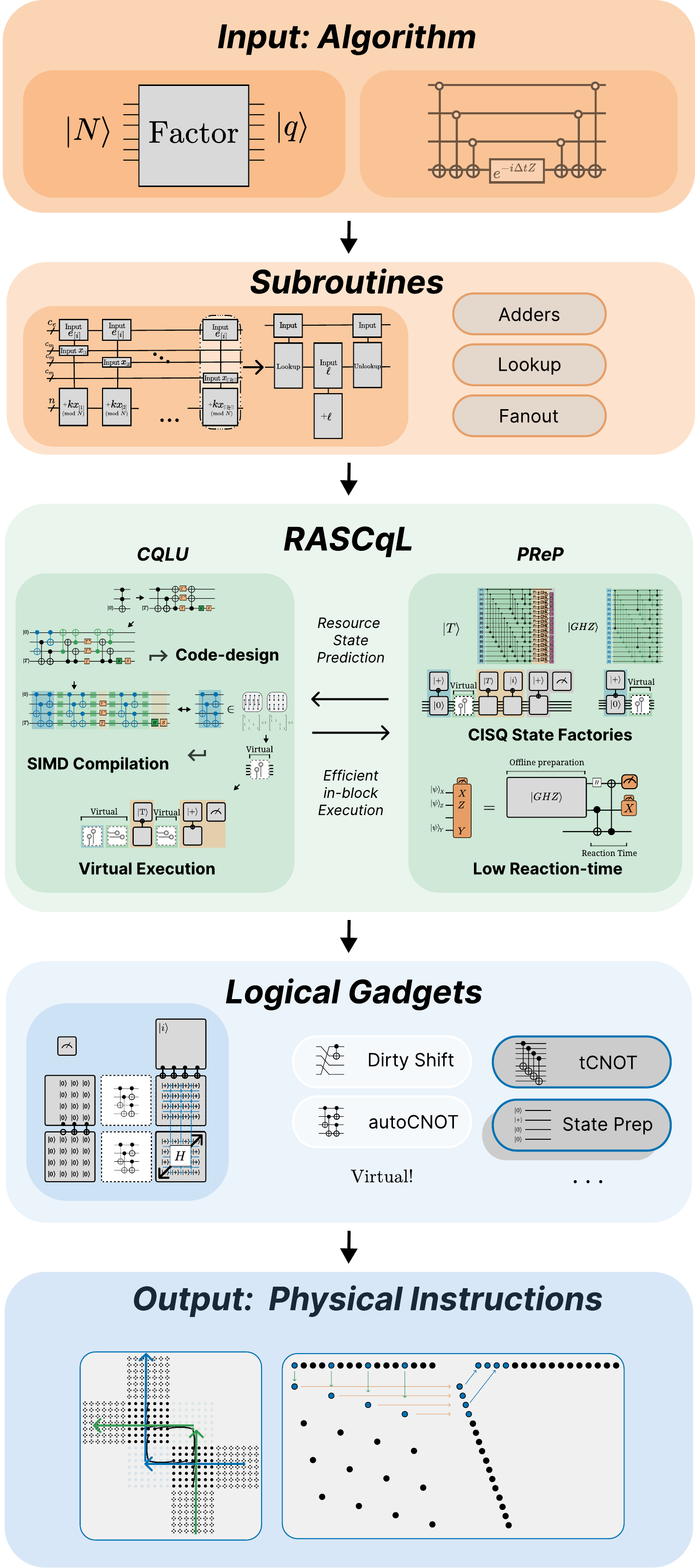}
    \caption{RASCqL Overview. Input algorithms supported by key subroutines are implemented using Complex Quantum Logic Units (CQLU) in tailored qLDPC codes and Predictive Resource-state Preparation (PReP) optimized using CQLU complex instructions. Logical primitives are then compiled down to parallel physical instructions on reconfigurable neutral atom arrays to obtain realistic estimations for logical performance and space-time cost.}
    \label{fig:rascql-eg}
\end{figure}

In this section, we describe RASCqL, a complex instruction set quantum (CISQ) computer architecture in qLDPC codes. First, an input quantum algorithm is compiled using a set of supported algorithmic subroutines. Two main subsystems enable efficient execution of these subroutines as complex instructions on qLDPC codes:  1). a suite of co-designed Complex Quantum Logic Units (CQLU) that realizes a limited but useful ISA in-block, and a tailored compilation of subroutines using this ISA; and 2). a Predictive Resource-State Preparation (PReP) protocol that allows efficient generation and low-latency consumption of resource states for non-native and reaction-time limited operations. Finally, CQLU and PReP primitives are implemented using parallel AOD movement and Rydberg interactions on RNAA architectures, which allows us to obtain RNAA-realistic logical error rate simulations and resource estimations. See Figure~\ref{fig:rascql-eg} for an abstract overview. 

\subsection{CQLU Design} 
First, we identify an ISA consisting of few low-order Clifford operations that a) is efficiently supported on some qLDPC codes and b) efficiently supports all desired subroutines. Starting with generic codes and applications, we propose LDPC-preserving code modifications that allow us to embed a prescribed group of logical Clifford gates as code automorphisms of an arbitrary qLDPC code, as we will show in Section~\ref{ssec:methods-code}. This code modification introduces additional footprint in general that scales with both the size of the ISA and available symmetry of the starting code that should be co-optimized for given applications. 
%hello :))) hi :D
Our ISA design, detailed in Section~\ref{ssec:methods-isa}, co-optimizes compilation and code construction to target magic state distillation, factoring, and certain chemistry simulation tasks. Our lightweight ISA with few complex entangling instructions is sufficient, except for gates that vary with problem instance, such as reactive measurements or the CNOT fan-outs in quantum lookup, as we will show in Section~\ref{ssec:methods-comp}. These instance-specific gates are implemented using gate teleportation, which requires $O(1)$ reaction time with provisioned resource states. We choose a family of Hypergraph Product of Simplex (HGPS) codes as the starting code, whose available primitive instructions already contain useful automorphisms such as dirty permutations and certain CNOT operators that support this ISA with minimal overhead. We further highlight the $[[450, 32, 8]]$ codes that require no overhead at all for this ISA. 

\subsection{PReP Design}
Resource states serve two purposes in RASCqL. Firstly, magic states and stabilizer states such as $\ket{T},\ket{i}$ and $\ket{GHZ}$ support non-native gates including $T$ gates and instance-specific $CNOT$ fan-outs. The consumption schedule of these states is deterministic and may be prepared offline. Secondly, logical ancillae are used for \emph{reactive measurements}--addressable measurements in a basis that may depend on previous measurement outcomes--arising possibly from magic state injection or parallelization. The speed at which we can implement reactive measurements determines the program runtime and overall volume in the \emph{reaction-time limited} compilation framework, which is the current SoTA for algorithms such as factoring. 

PReP decouples state production from program execution using prediction and catalysts, reducing reaction time to $O(1)$ in expectation and unlocking additional resource savings from pipelining. First, PReP predictively provisions a pool of modified $\ket{GHZ}$ states that can implement reactive measurements in $O(1)$ time \cite{HarnessingGHZ}. As we will show in Section~\ref{ssec:reactiveops}, only a single type of states needs to be prepared if we allow additional constant time CQLU operations on the fly. PReP also maintains a generating set of $\ket{i}$ states on diagonal locations that serve as catalysts to perform arbitrary patterns of $X/Y$-basis reactive measurements which arise in MSD circuits. Finally, all resource states are generated in qLDPC blocks with optimized CQLU protocols, unlocking additional gains from high-rate magic state distillation and pipelined state production protocols. 

\subsection{RNAA Implementation}
RASCqL is also equipped with a physical compiler to implement QEC primitives efficiently on an RNAA device. Memory cycles and complex instructions on CQLUs are compiled to explicit AOD movements, gate schedules, and mapping, which determines the real-time costs of each logical instruction. RASCqL's physical scheduler uses fast and parallel systolic movements to achieve complex qubit permutations and memory cycle times on a millisecond scale, exploiting the algebraic structure of the codes and transversal logic. Details are given in Section~\ref{ssec:methods-rnaa}. A simple greedy mapper is employed to minimize interaction distance, which can also be adapted for a zoned architecture.

Furthermore, an accurate circuit level noise model is used to assess the logical error rate performance of the proposed code that factors in additional gate errors and idle time due to transversal logic. For a neutral atom implementation of HGPS codes, we observe a high circuit level threshold of $0.78\%$, as shown in Section~\ref{ssec:eval-ler}. We use the fits obtained from these simulations to estimate system-level costs of logical subroutines.

\section{Methods}\label{sec:methods}

\subsection{CQLU Construction}  \label{ssec:methods-code}
Given an error correcting code, it is well understood how to determine the set of logical operators arising from its code automorphisms\cite{berthusen2025automorphismgadgetshomologicalproduct,danbrownauto}. However, given an efficient ISA for a class of applications, the question of how we can construct a QECC containing specific logical operations is both less explored and more practical to consider. Our CQLU design makes two technical contributions to the latter question: first, we adapt code modification techniques such as augmentation and extension for the purpose of embedding a logical operation as matrix automorphisms; secondly, we identify a class of highly symmetric HGPS codes that require little modifications to achieve an efficient ISA. We summarize the codes and logical ISA in Table~\ref{tab:codes}, Table~\ref{tab:isa} and Figure~\ref{fig:CQLU-isa} while detailed proofs are deferred to Appendix~\ref{appx:proofs}.

We will introduce a motivating example. Formally, given a binary linear code $\mathcal{C}$ with generator matrix $G$ and parity check matrix $H$, an automorphism on $\mathcal{C}$ refers to a coordinate permutation $\mathcal{\sigma}$ that preserves the space of all code words. That is, there exists some invertible linear transformation $g$, such that \begin{equation}\label{eq:auto-cond}
    gG\sigma = G.
\end{equation}
Equation~\ref{eq:auto-cond} is called the \emph{automorphism condition}. To enact the induced logical action of $g$ on the classical code $\mathcal{C}$, one can virtually permute coordinates, and modify subsequent encoding/decoding. When using $\mathcal{C}$ to construct a quantum code such as the Hypergraph Product code, a stricter condition of \emph{matrix automorphism} is necessary for automorphisms $(g,\sigma)$ to realize a permutation-transversal logical operation\cite{berthusen2025automorphismgadgetshomologicalproduct}, where 
\begin{equation}
    H\sigma = \rho H
\end{equation}
for some row permutations $\rho$. When $\sigma$ is a matrix automorphism for some classical code $\mathcal C$, the induced logical action lifts to a fault-tolerant logical operations $g\otimes I$, $I\otimes g$, or $g\otimes g$ on the HGP of $\mathcal{C}$, all of which are implementable as virtual qubit-relabeling. 

We do not expect codes to contain non-trivial automorphisms in general. For example, consider the smallest binary error correcting code with two logical bits--a $[5,2,3]$ shortened Hamming code--with generator and check matrices:
\begin{align}
    G & = \begin{bmatrix}
        1&&1&1&\\
        &1&&1&1
    \end{bmatrix},\label{eq:eg-G}\\
    H &=\begin{bmatrix}
        1&&1&&\\
        1&1&&1&\\
        &1&&&1
    \end{bmatrix}.
\end{align}
The SWAP between the two logical qubits is an automorphism-induced gate of $G$. By physically swapping the first two physical bits, and then the third and fifth bits, we obtain a new basis $G_{SWAP}$ that corresponds to swapping the two rows of $G$. However, performing a CNOT between the two encoded bits (which corresponds to applying $g = \begin{bmatrix}
    1&\\1&1
\end{bmatrix},$) would yield 
\begin{equation}\label{eq:eg-gG}
    gG = \begin{bmatrix}
        1&&1&1&\\
        1&1&1&&1
    \end{bmatrix} 
\end{equation}
that clearly cannot be obtained from $G$ via a column permutation since it contains a weight-$4$ row.

Suppose we are given a logical action $g$ and a code $G$, can we obtain a related code $G'$ that implements $g$ as a coordinate permutation? In this case, it suffices to add a single column $\begin{bmatrix}
    1\\1
\end{bmatrix}$, which yields

\begin{align}
    gG' &= \begin{bmatrix}
    1&\\1&1
\end{bmatrix}\cdot \begin{bmatrix}
        1&&1&1&&1\\
        &1&&1&1&1
    \end{bmatrix}\\
    &= \begin{bmatrix}
        1&&1&1&&1\\
        1&1&1&&1&
    \end{bmatrix} \\
    &= G'\cdot\sigma_P
\end{align}
where $\sigma_P$ is a column permutation that swaps the first three columns with the last three columns. That is, an augmentation on $G$ yields a $[6,2,4]$ code $G'$ whose automorphism group is expanded to include an additional logical gadget $g$.

We can prove a loose upper-bound for a generic code:

\begin{theorem}[Code Automorphism Expansion]
    Let us be given a $[n,k,d]$ code $\mathcal{C}(G,H)$ and $\mathcal{G} = \{g_1,...,g_m\} \leq \gln{k}$. There exists a family of $[n'\leq nm,k,d'\leq dm]$ codes $\mathcal{C}_{\mathcal{G}}(G_{\mathcal{G}},H_{\mathcal{G}})$ such that $\mathcal{C}\cong\mathcal{C}_\mathcal{G}$, $\mathcal{G}\leq \mathcal{L}_{G_\mathcal{G}}({\mathcal{C}_\mathcal{G}})$. Furthermore,  if $H$  is $w$-bounded, $g$ is $t$-bounded for all $g\in \mathcal{G}$ and $|\mathcal{G}| = m$,  there exists a $(w + t+ 1)$-bounded check matrix $H^{(t)}_\mathcal{G}$ and a $(w + m+ 1)$-bounded check matrix $H^{(m)}_\mathcal{G}$.
\end{theorem}
\begin{theorem}[qLDPC with Virtual Instruction]
Let us be given $\mathcal{C}(G,H)$ and $\mathcal{G} \leq \gln{k}$. Suppose $H$ is $w$-bounded, $g$ is $t$-bounded for all $g\in \mathcal{G}$ and $|\mathcal{G}| = m$. Then, there exists a family of $[[nm, k, dm]]$ codes $\mathcal{C}_\mathcal{G}(G_\mathcal{G},H_\mathcal{G})$ with a $2(w +mt + 1)$-bounded check matrix $H_\mathcal{G}$ such that  $g\in \mathcal{G}$ can be implemented as a virtual qubit relabeling.
\end{theorem}

%This code modification scheme may not be practical in general. At the limiting case, embedding $\gln{k}$ to any classical code yields the classical Simplex code, and a lifted quantum code with weight proportional to code length. However, when the code and ISA is appropriately chosen, we can retain both low-overhead and LDPC properties. T
The detailed proof is given in Appendix~\ref{appx:proofs}, which also includes some useful scenarios of when such code modifications incur low overheads. 

\subsection{CQLU Logical Primitives}\label{ssec:methods-isa}
While our code modification techniques apply generally to all codes, its performance and whether we have the LDPC guarantee depend on the ISA we wish to embed and the starting classical code, which should be co-optimized with a chosen application. 

In this work, we highlight a family of hypergraph product (HGP) codes, $\mathcal{Q}_r$, with classical Simplex codes $\mathcal{S}_r$ as seed codes. Since the classical simplex code contains all reversible linear operations as code automorphisms\cite{ComputingQLDPC}, no additional overhead is required for automorphism expansion. There also exists a weight-$3$ circulant parity check matrix (pcm) that contains a set of logical cyclic shifts with additional CNOTs from the qubit shifted (see Figure~\ref{fig:CQLU-isa}.b) as matrix automorphisms. We call these shifts dirty cyclic shifts. In an HGP construction, these matrix automorphisms can either be used to perform CNOTs between specific qubits in fixed positions, or a cyclic shift provided that either the control qubit is set to $\ket{0}$ or the target qubit is set to $\ket{+}$. 

Additionally, a 4-bit CNOT circuit $g_{auto}$ (see Figure~\ref{fig:CQLU-isa}.a) is embedded as a matrix automorphism. Since $g_{auto}^2=I$, the check extension incurs at most $2\times$ overhead. In particular, for the $[[450, 32, 8]]$ code $\mathcal{Q}_4$, there exists a pcm that contains both the dirty shifts and $g_{auto}$ as matrix automorphisms with no additional overhead. $g_{auto}$ lifts to $16$ qubit CNOT circuits in an HGP construction that correspond to $g_{auto}\otimes I$, $g_{auto}$ applied to each column; $I\otimes g_{auto}$, $g_{auto}$ applied to each row; or $g_{auto}\otimes g_{auto}$, as illustrated in Figure~\ref{fig:CQLU-isa}.a. The full ISA is summarized in Figure~\ref{fig:CQLU-isa} and Table~\ref{tab:isa}. 

\begin{table}
    \centering
    \begin{tabular}{c|cc}
          Codes ($[[n,k,d]]$)&  Footprint ($\frac{n+m}{k}$) &Cycle Time\\\hline
           $[[450, 32, 8]]$&   $29.2\times $ & $3.9~ms$\\
 $[[90,8,10 (\leq 8)]]$& $22.5\times$ & $6.5~ms$\\
  $[[49, 1,7]]$& $98\times $ &$1.4~ms$\\
 $[[81,1,9]]$& $162\times$ &$1.4~ms$\\
    \end{tabular}
    \caption{Comparison of space- and time-overheads of $\mathcal{Q}_4$, a Hypergraph-Product of Simplex code with the rotated surface codes \cite{fowler} and bivariate bicycle codes \cite{IBMBB}. Cycle times are calculated using systolic scheduling \cite{JoshBB} on a RNAA device. The improved cycle time is due to $\mathcal{Q}_4$ having a simple cycle structure, where $X/Z$ checks are only shifted in one direction, as opposed to two--which is the case for the $[[90,8,10 (\leq 8)]]$ bivariate bicycle code. }
    \label{tab:codes}
\end{table}

\begin{figure*}
    \centering
    \includegraphics[width=0.9\linewidth]{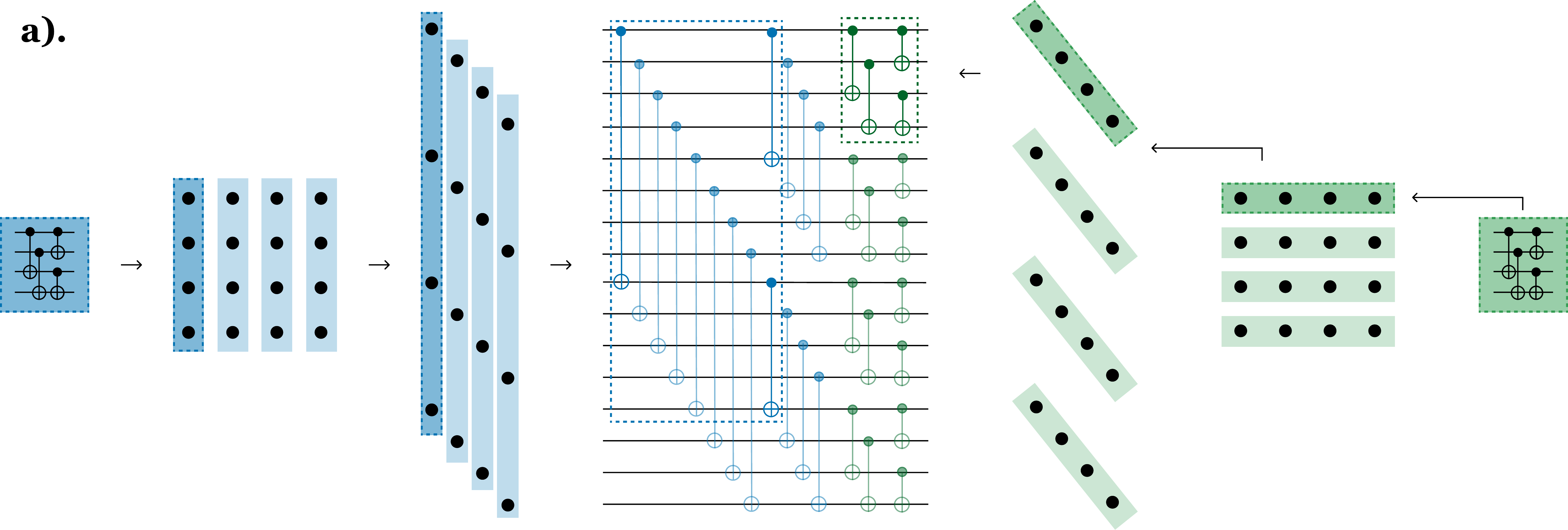}
    \includegraphics[width=0.95\linewidth]{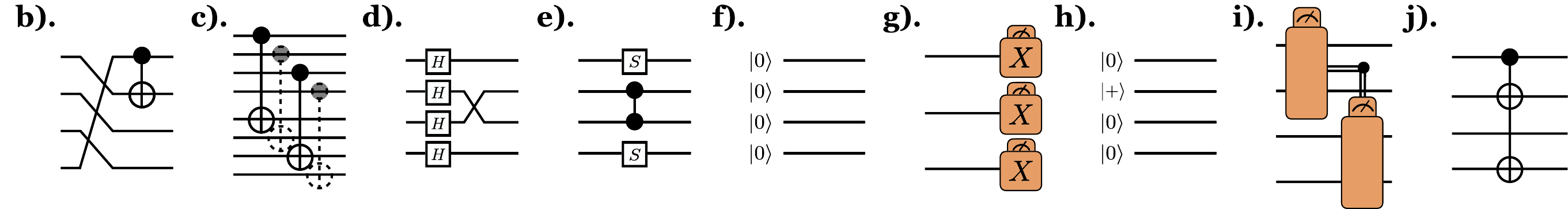}
    \caption{Set of logical instructions on CQLUs}
    \label{fig:CQLU-isa}
\end{figure*}
\begin{table*}
    \centering
    \begin{tabular}{c|cccl}
         Instructions& Total Time& Logical Ancilla&Reaction Time  &Execution Time\\\hline
         a). Automorphism CNOTs (autoCNOTs)&  0& 0&0  &$0$ \\
         b). Dirty Cyclic Shifts&  0& 0&0  &$0$ \\
 c). Transversal CNOTs (tCNOTs)& $O(1)$& 0&$O(1)$  &$3.0~ms$ ($+2.4~ms$)\\
 d). $H$-SWAP& $O(1)$& 0&$O(1)$  &$1.9~ms$\\
 e). $CZ$-$S$& $O(1)$& 0& $O(1)$  &$1.9~ms$\\
 f). $X^{\otimes k}/Z^{\otimes k}$ Measurements& $O(1)$& 0&$O(1)$  &$0.5~ms$\\
 g). $\ket{0}^{\otimes k}$, $\ket{+}^{\otimes k}$& $O(d)$& 0&$O(1)$  &0 ($31.2~ms$)\\
 h). Mixed $\ket{0}$, $\ket{+}$ States& $O(d+\sqrt k)$& $O(\sqrt k)$&$O(1)$  &$6.0~ms$ ($+31.7~ms$)\\
 i). Reactive PPM& $O(d+\sqrt k)$& $O(\sqrt k)$&$O(1)$  &$6.0~ms$ ($+44.2~ms$)\\
 j). Arbitrary Fanouts& $O(d+\sqrt k)$& $O(\sqrt k)$&$O(1)$  &$6.0~ms$ ($+38.2~ms$)\\
    \end{tabular}
    \caption{CQLU ISA and costs. Gates a). and b). can be implemented virtually via qubit relabeling, which may incur additional $2.4~ms$ execution time for the next tCNOT. For gates g).--i)., we assume an logical ancilla is prepared offline. The total execution time without provisioning is reported in parentheses. }
    \label{tab:isa}
\end{table*}

\subsection{Logical-Functional Compilation}\label{ssec:methods-comp}
\begin{figure*}
    \centering
    \includegraphics[width=0.95\linewidth]{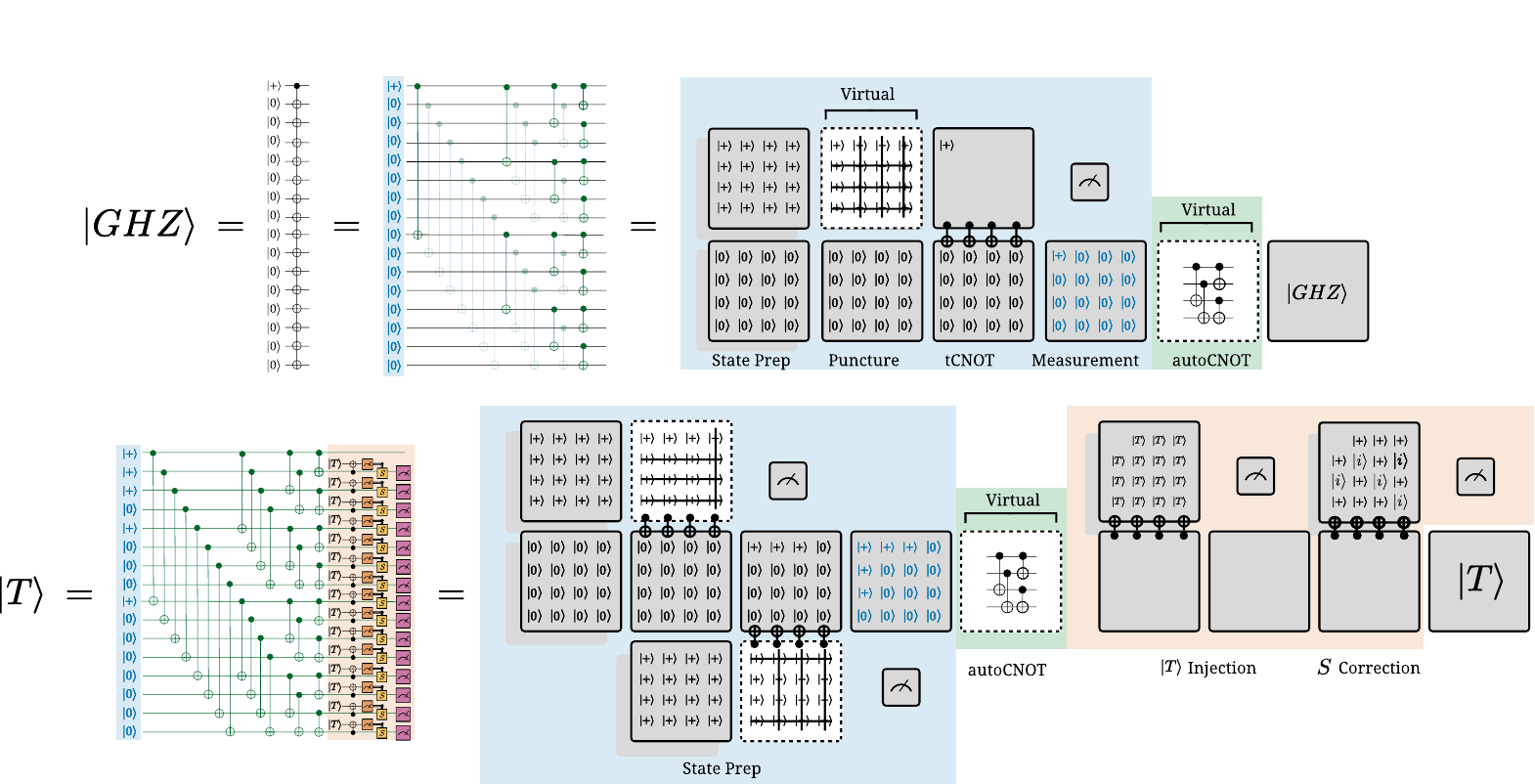}
    \caption{State preparation circuits using CQLU instructions. The initial Pauli eigenstate state preparation (highlighted in blue) can be done by preparing two (in the case of GHZ) or three (in the case of $[15,1,3]$) patches in $X$ or $Z$ eigenstates followed by puncturing and homomorphic measurements\cite{hgphomo}. The entangling operations (highlighted in green) can be done using virtual embedded CNOTs, where the shaded CNOTs in the GHZ circuit act trivially. Magic state injection (highlighted in gray) is performed using a non-native LPU. The Clifford correction (highlighted in orange) is implemented as a mixed $X/Y$ reactive measurement.}
    \label{fig:state-prep}
\end{figure*}
\subsubsection{State Preparation}
We can use $g_{auto}$ to prepare a $GHZ$ state on all qubits in the same block, as illustrated in Figure~\ref{fig:state-prep}. We can grow the GHZ state across multiple patches using transversal CNOTs between a $\ket{GHZ}$ patch and $\ket{0}^{\otimes k}$ patches. To prepare a GHZ state with arbitrary support, we can first entangle all qubits, and then selectively measure out qubits not in the support by preparing a mixed $\ket{0}/\ket{+}$ state in the corresponding pattern.

For MSD circuits, the only non-native operation required is the magic state injection, as illustrated in Figure~\ref{fig:state-prep}, which dominates resource costs. We expect that as better schemes are developed for magic state injection, RASCqL's performance may improve in tandem. 

\subsubsection{Quantum Arithmetic}
We use Gidney's ripple carry adder that uses $4n+O(1)$ $T$ gates for an $n$-bit adder\cite{Gidney_2018}, where each MAJ and UMA block is compiled using time optimal methods \cite{fowler2013timeoptimalquantumcomputation}. We then re-write this circuit so the bulk of the Clifford gates can be executed in the CQLU ISA. In particular, the 3-bit fan-outs and fan-ins within the MAJ block are realized using repeated applications of $g_{auto}$ and global $H$, while the Bell-preparation and Bell-measurement can be realized using tCNOTs by storing the bridge qubits in a separate qLDPC block. The Pauli corrections from Bell-measurements and Clifford corrections from $\ket{T}$ injections are propagated past the $Z$-basis measurements, turning them into reactive measurements with basis depending on previous measurement results. In the end, all operations in the gray box are executed in parallel using cheap in-block operations, followed by reactive measurements that are resolved in series with $O(1)$ delay.

The quantum look-up tables can be implemented by combining Toffoli ladders (which appear as a subcircuit of the MAJ block) and CNOT-fanouts (which can be implemented by injecting a GHZ state).

\begin{figure*}
    \centering
    \includegraphics[width=0.95\linewidth]{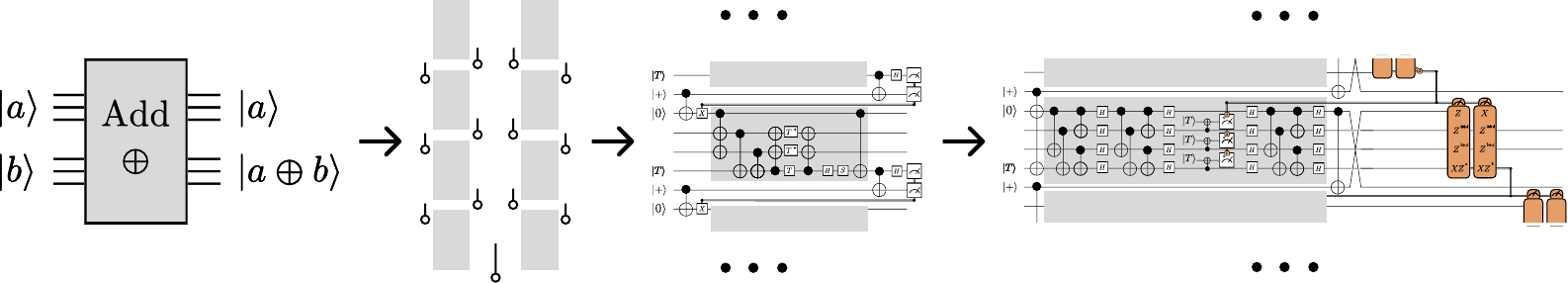}
    \caption{Gidney's ripple carry adder\cite{Gidney_2018} compiled using CQLU instructions. Eight MAJs can be implemented in each block with the same time cost. See Appendix~\ref{appx:circuits} for a detailed derivation. }
    \label{fig:adder}
\end{figure*}

\subsection{Logical-Physical Compilation}\label{ssec:methods-rnaa}
In this section, we describe our hardware assumptions, and how we achieve realistic estimates for the time costs of CQLU primitives. We also describe explicit layouts and schedules to implement non-trivial permutations exploiting the parallel atom re-arrangements in RNAA platforms.

\subsubsection{Hardware Assumptions}
For our estimates we presume a RNAA device that supports both mid-circuit, parallel atom movement via AODs and in-place measurement of ancilla qubits. High fidelity atom movement has been demonstrated in experiment~\cite{Bluvstein2024}, leading us to assume moving distance $l$ requires time $2\sqrt{l / a}$ for an acceleration of $a=5500~m/s^2$, consistent with prior work on QEC architectures in RNAA~\cite{Zhou_2025}. In-place measurement of ancilla qubits has also been demonstrated in experiment in dual-species~\cite{anand2024dual} neutral atom devices. We note that RASCqL is not reliant on in-place measurement and is also compatible with a zoned architecture~\cite{Bluvstein2024} where ancilla qubits are moved to a separate zone for measurement.

%\subsubsection{Algebraic Structures} 
%\todo{Product structure, cyclic shifts for syndrome extraction, systolic movement schedules}
%\subsubsection{Logical ISA}
\begin{figure}
    \centering
    \includegraphics[width=\linewidth]{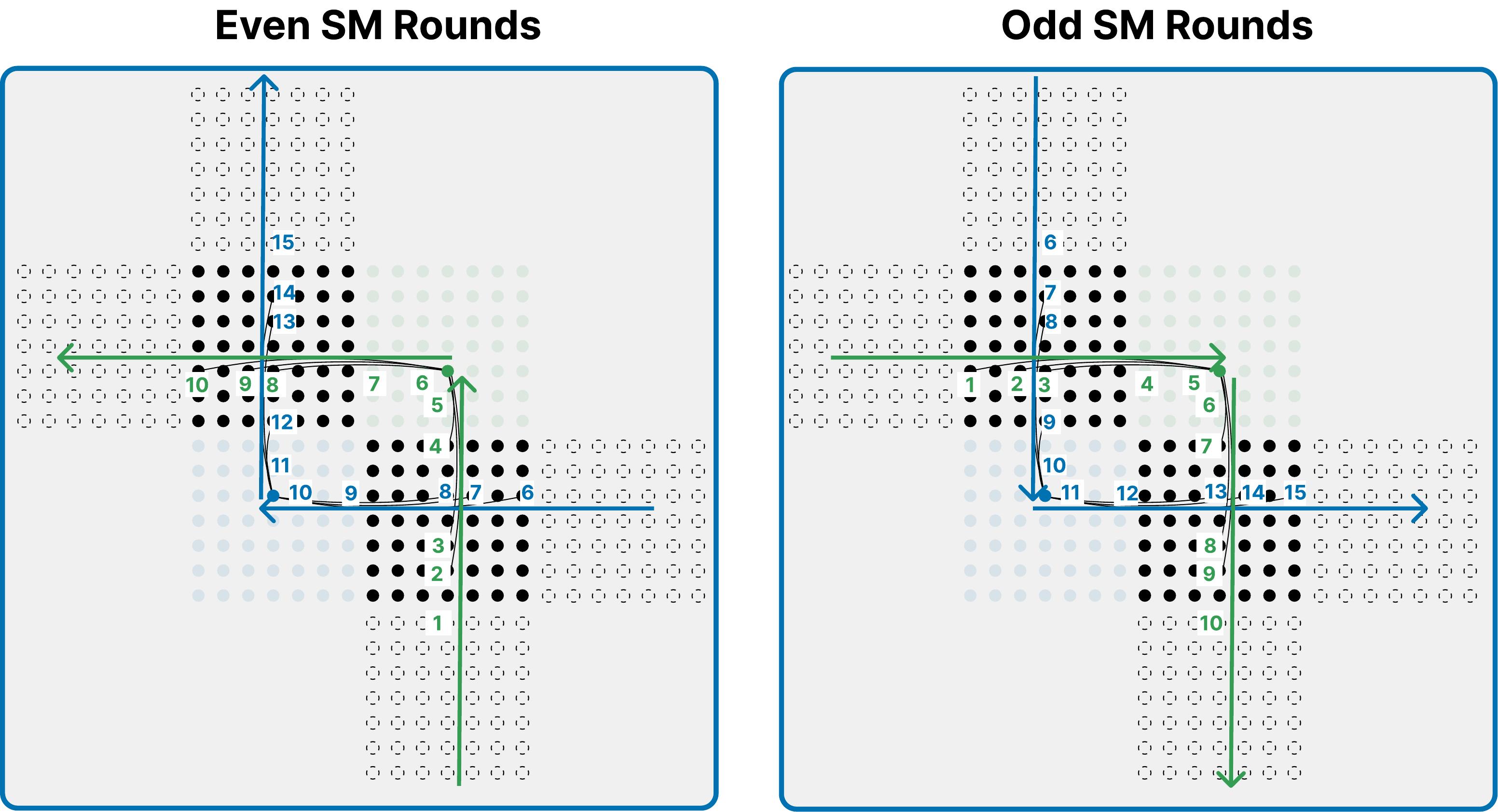}
    \includegraphics[width=0.75\linewidth]{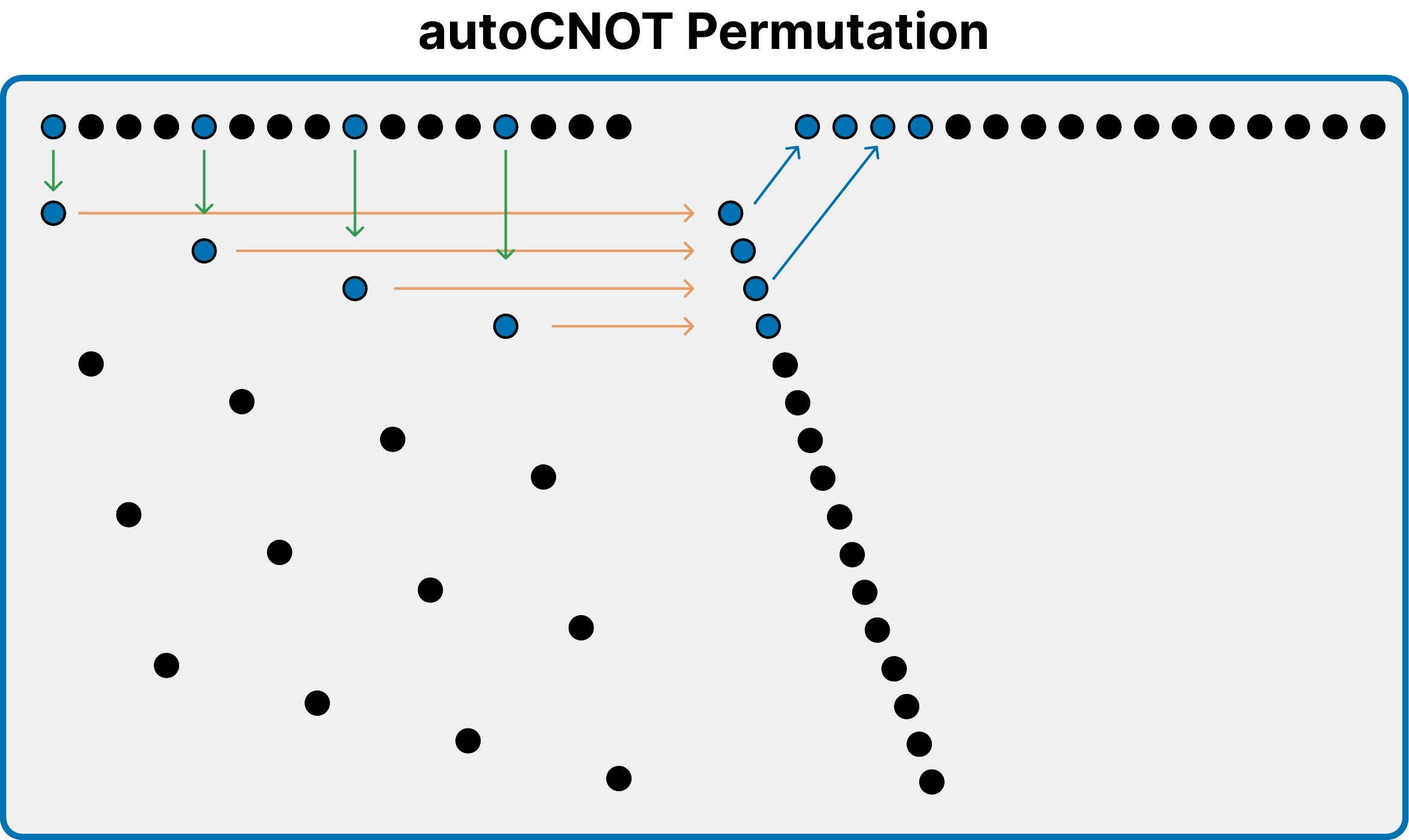}
    \caption{Top: Alternating movement schedule for syndrome measurement (SM) of the codes in RASCqL. X/Z check qubit arrays (green/blue) are moved collectively along the defined path. Numbered stops for parallel 2Q gates indicate the position of an example X/Z check qubit. Bottom: Three AOD movements (green, orange, blue) that implements the qubit permutation for autoCNOTs, which only needs to be executed before tCNOTs. }
    \label{fig:sm_movement}
\end{figure}
\subsubsection{Syndrome Extraction}
The qLDPC codes described in this work are HGPs of the classical Simplex codes, which are themselves constructed from cyclic permutation matrices.
The codes are therefore closely related to Generalized-Bicycle codes and benefit from the same movement parallelism exploited in~\cite{JoshBB}.
Based on this insight we create an efficient movement schedule for reconfigurable neutral atom devices detailed in Figure~\ref{fig:sm_movement}.
Since all CNOT orderings for HGP codes are distance-preserving~\cite{manes2025distance}, our schedule uses an ordering that minimizes the required movement.
To improve movement times, we break the schedule into alternating rounds.
Importantly, in each schedule the X stabilizers (green) interact with each data qubit first, ensuring all stabilizers commute.

\subsubsection{Virtual Logical Gates}
Performing virtual logic gates in RASCqL's ISA only requires a permutation of the physical data qubits that can usually be implemented as a virtual relabeling. 

One concern is that such qubit relabelings may disturb existing structures our optimized syndrome extraction schedule exploits, and cause future QEC cycles to be slower. We show that this is not the case for the HGPS codes. First, the Dirty Cyclic Shifts operation requires a cyclic shift of the data qubits. Since HGPS is also cyclic, it is easy to see that this shift only modifies the ordering of future syndrome extraction cycles, but does not change the overall timing.  The atom movement required for the autoCNOTs is a fold symmetry on the columns (rows). As we we show in Proposition C.2, this symmetry preserves the stabilizers exactly, hence requiring no changes to future syndrome extraction cycles.

We may need to implement the permutation physically when performing a tCNOT that involves a relabeled patch. We can also show that when combined with atom rearrangements for the tCNOT operations, permutations from virtual relabeling can be implemented efficiently. For the Dirty Cyclic Shifts, it suffices to move one column (row) before a tCNOT, leading to at most $1.2~ms$ delay. For autoCNOTs, three additional movements are needed to rearrange the grid. We can describe it for the first qubit in each column (row), which we label with two indices $x,y$ such that qubit $i = 4x+y$, and suppose the initial position of qubit $i$ are $(i,0)$. The three movements are $(4a+b,0)\mapsto (4a+b,4b+a)\mapsto((4a+b)/16, 4b+a)\mapsto (4b+a, 0)$, where the last step can be combined with the $tCNOT$. See Figure $9$ for a visual illustration. 

% HGPC Cyclic -> 2D Systolic Checks (Code)
% AOD Movement -> 2D Systolic (Hardware)

% Virtual CNOTs -> Qubit Relabeling 
% - Fig. 6a) Fold-transversal
% - Fig. 6b) Cyclic shift (2-step) rows for L block, cols for R block

%\subsubsection{Zoned Archiecture} 
%\todo{Discuss sensitivity to inplace measurement/zoned architecture, hardware params}

\subsection{Resource State Management}

\subsubsection{State-mediated Reactive Operations}\label{ssec:reactiveops}
\begin{figure}
    \centering
    \includegraphics[width=0.95\linewidth]{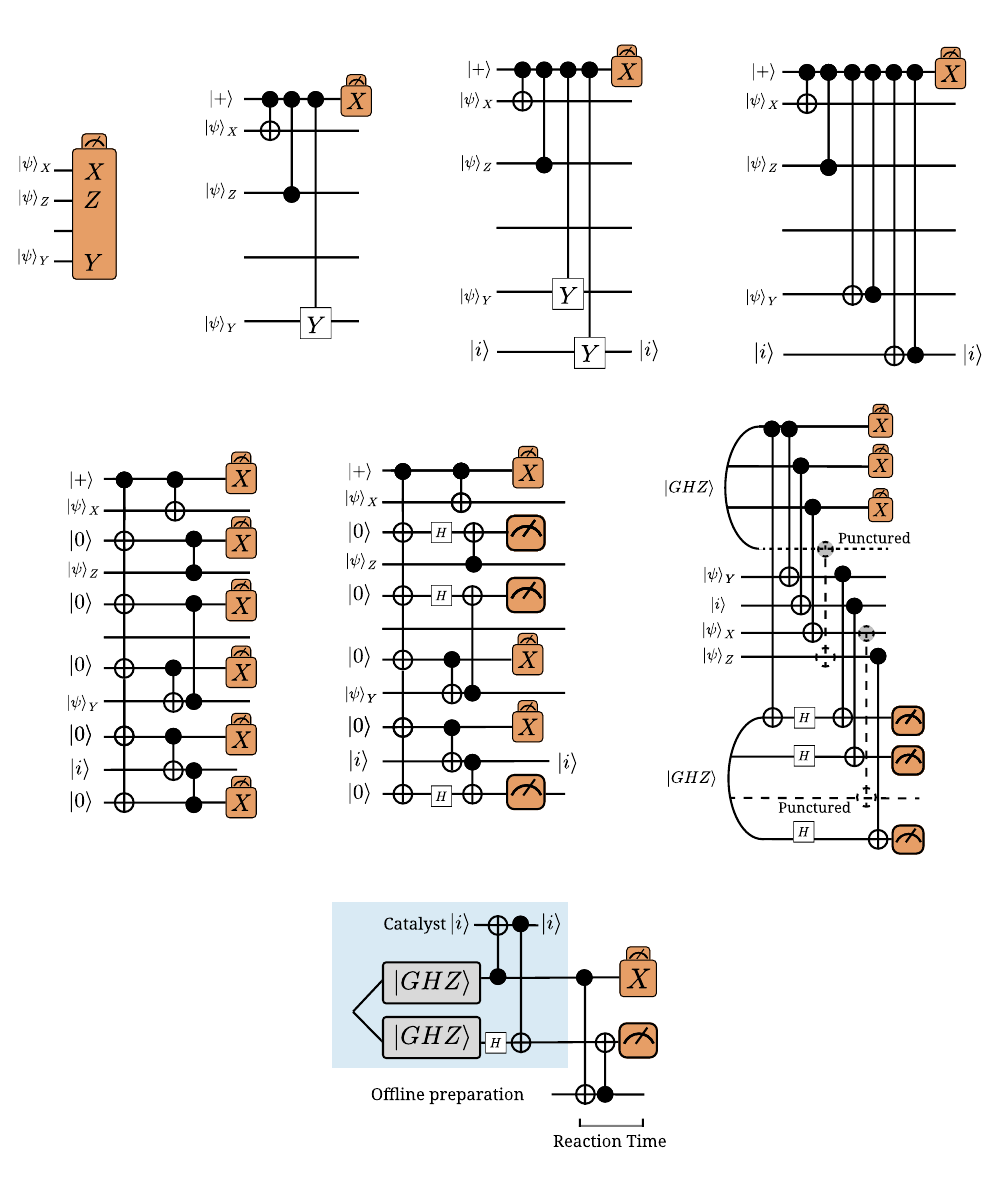}
    \caption{Preparation and consumption of a catalyzed $GHZ$ state for reactive measurements. When a catalyzed $GHZ$ state is available, it takes 2 tCNOTs ($6~ms$) to implement a reactive measurement. Online preparation, on the other hand, requires an additional $38.2~ms$ for $\ket{GHZ}$ preparation and $6~ms$ to use the catalyst.}
    \label{fig:reactivemeas}
\end{figure}
For adders and quantum look-up tables, Toffoli gates are implemented using $\ket{T}$ state injections, which may require $S$ corrections. The reason we use $\ket{T}$ instead of $\ket{CCZ}$ states is that it's easier to implement $S$ corrections using $\ket{i}$ state catalysts than non-native $CZ$ operations. 
These corrections may also depend on $X$ errors from using Bell-pairs for parallelization. Rather than enforcing the correct measurement dependencies and implementing $S$-corrections on-the-fly, we can commute potential Clifford corrections to the end of the circuit, resulting in a weight-$4$ Pauli product measurement as shown in the final step of Figure~\ref{fig:adder}. These reactive measurements can be implemented with two transversal CNOTs by consuming a special $\ket{GHZ}$ state (possibly prepared offline), as shown in Figure~\ref{fig:reactivemeas}.

\subsubsection{Reducing Provisioning Size}
For the $[[450,32,8]]$ code, $8$ MAJ blocks or temporary-AND Toffolis can be implemented in parallel. Each block may require reactive measurements parameterized by $5$ measurement outcomes--two from the Bell-measurement of the previous MAJ block, and three from $\ket{T}$ injections--giving rise to $8\times 2^5 = 256$ potential resource states to be requested. We reduce this number to a single type by allowing additional CQLU primitives. 

Firstly, we arrange the logical qubit placement in CQLUs so the support of each reactive measurement is in the same row or column. We prepare logical ancillae for only particular columns embedded in a block with $\ket{0}$ states, and use dirty cyclic shift to access the correct state. The layout may incur additional overhead for the logical computation, since bridge qubits used for parallelization participate in operations that belong to multiple blocks. Instead of assigning particular rows/columns to bridge qubits, we instead store them separately and SWAP them into the corresponding location before measurements, which can be done transversely and in parallel. We can again use dirty cyclic shifts for re-arrangements when we embed the empty rows/columns with $\ket{0}$ eigenstates. Since any additional operation are done before measurements take place, it does not increase reaction time, and while it does increase the qubit footprint, the cost is not significant over the whole program due to the short lifespan of the additional qubits needed.

Secondly, we can split the measurement into $X$-and $Z$-terms and inject using two patches (where $Y = iXZ$), as illustrated in Figure~\ref{fig:reactivemeas}. We can assume each qubit requires an $X$-term and a $Z$-term, and use virtual puncturing operations depending on the measurement results. This increases the reaction time by an additional transversal CNOT, but reduces the pool of potential states to $1$. There is a caveat, however, when we use global transversal CNOTs for state-preparation: the $Y$ terms require using a $\ket{i}$ catalyst, which requires a transversal CNOT and a transversal CZ onto the same qubit. While the $CZ$ can be implemented as a CNOT conjugated by Hadamards, it also enacts a logical transpose on the encoded qubits, so we can only apply $CNOT_{i,j}$ and $CZ_{i,j}$ between qubit $i,j$ in two blocks iff $i = j$. We overcome this by observing that each reactive measurement contains at most one $Y$-term; when we SWAP the bridge qubits into the $i$-th row/column, we can make sure that the $Y$ term is in the $i$-th position, so the targets of both the CNOT and CZ is a qubit on the diagonal. 

\subsection{Scheduling and Pipelining State Generation}
PReP allows us to decouple state generation from program execution with a small provision that contains catalyst $\ket{i}$ states on diagonal positions, and a pool of consumable $\ket{T}$ states and modified $GHZ$ states. Another benefit of off-line preparation is now we can produce such states in batches with higher efficiency. For example, since GHZ states are prepared on only a single row/column, we can batch the preparations of $8$ reactive-measurement ancillae together to save space-time cost. Another example is for magic states: rather than implanting two $[[15,1,3]]$ MSD in parallel in each block, we can use a $[[109, 19, 3]]$ over $4$ blocks that require two additional global $tCNOT$ for a $2.6\times$ increase in per-logical state efficiency. 

Finally, since request rate and amount is available prior to execution, we can in principle optimally schedule state generations per performance targets and device constraints. To optimize fidelity or total active volume, we may generate states with an \emph{as late as possible} policy, while the additional qubits required for state generation may drastically increase program footprint if request rates spike. On qubit-constrained devices, we may instead generate resource states at a constant rate to match program requirements on average; while able to maintain low-footprint, it may be required to halt program executions for resource states in certain request patterns.  In this work, we adopt the latter strategy to retain the space advantage and assume the overhead, and leave more detailed optimization for future work.

\section{Evaluation}\label{sec:eval}

\subsection{Circuit level simulations}\label{ssec:eval-ler}

To estimate the performance and cost of RASCqL, we perform circuit-level simulations using Stim \cite{Gidney_2021} and decode using BP-OSD. We use a standard circuit-level noise model where initialization and measurement experience bit-flip errors with probability $p$ and one- and two-qubit gates cause depolarizing errors with probability $p$. To model the extra error introduced by movement-based automorphism gates and transversal operations, we insert an extra depolarizing error with probability $p$ on the data qubits halfway through the experiment. Figure~\ref{fig:hyps-memory} shows the logical memory performance of HGPS codes for $r \in [3, 4, 5]$. We fit this data to the expect scaling behavior $p_L = A(bp)^{(d+1)/2}$, obtaining $A=0.221 \pm 0.048$ and $b=128.34 \pm 0.20$. Note an empirical threshold $p_{th}$ is observed to be $0.78\%$. In comparison, we expect SHYPS (shown in Figure~\ref{fig:shyps-memory}) to lack a clear threshold due to the increasing effect of measurement errors as code distance increases. We therefore fit the SHYPS data to the modified scaling function $p_L = A((b+cd)p)^{(d+1)/2}$, obtaining fit parameters $A=0.1497$, $b=0.3894$, and $c=50.83$. We plot our fits for HGPS, SHYPS, Surface codes and compare them against those reported in \cite{IBMBB} in Figure~\ref{fig:lercomp}.  
\begin{figure}
    \centering
    \includegraphics[width=0.8\linewidth]{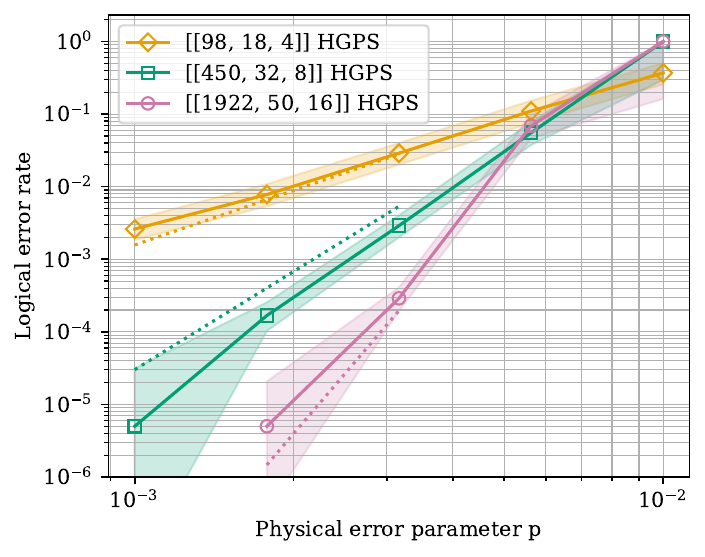}
    \caption{Logical Error Rate of HGPS under neutral atom inspire circuit level simulation. We account for additional idle and gate errors from transversal logical operations by inserting extra noise mechanisms in the middle of QEC cycles. We observe a threshold of around $0.78\%$.}
    \label{fig:hyps-memory}
\end{figure}
\begin{figure}
    \centering
    \includegraphics[width=0.48\linewidth]{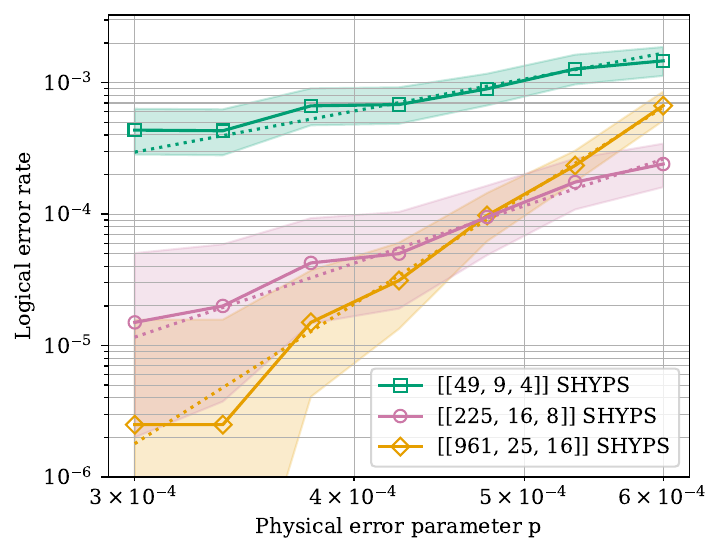}
    \includegraphics[width=0.46\linewidth]{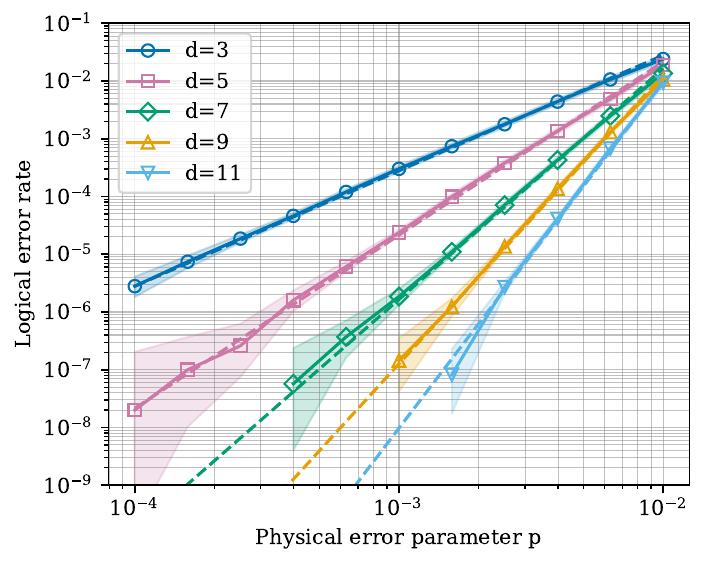}
    \caption{Logical Error Rate of SHYPS and Surface codes under neutral atom inspire circuit level simulation.}
    \label{fig:shyps-memory}
\end{figure}

\begin{figure}
    \centering
    \includegraphics[width=0.72\linewidth]{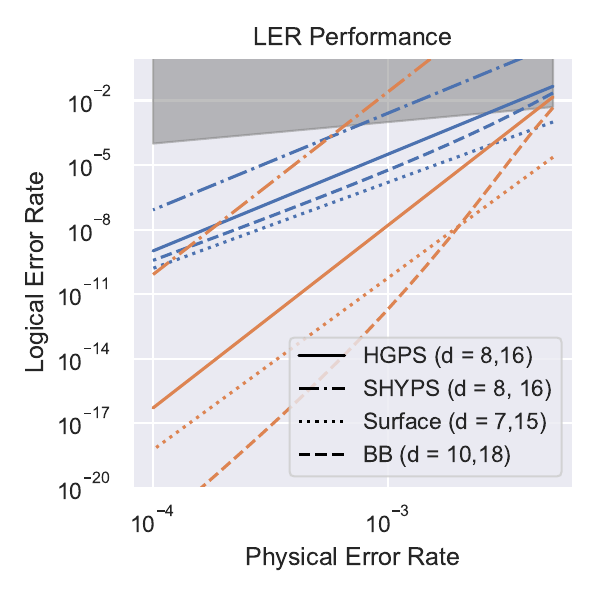}
    \includegraphics[width=\linewidth]{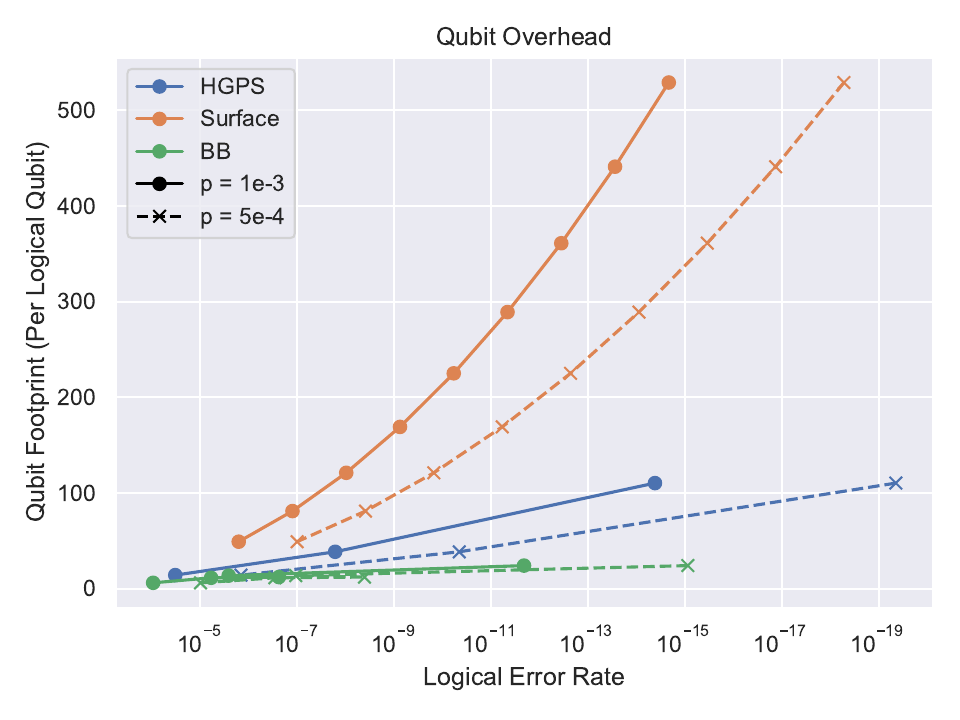}
    \caption{Top: Comparison between LER performance across HGPS (what our CQLU uses); SHYPS, a compute efficient subsystem code construction of the Simplex code \cite{ComputingQLDPC}; the rotated surface code (our baseline)\cite{Fowler2012surface}; and Bivariate Bicycle code, hardware compatible qLDPC code with bi-planar checks\cite{IBMBB}. Bottom: the projected per-logical qubit footprint for a range of target logical error rates. }
    \label{fig:lercomp}
\end{figure}

Finally, we simulate the injection of a quantum state from a $d=3$ surface code into $r=3$ and $r=4$ HGPS codes using the universal adapter construction\cite{UniversalAdaptersIBM}, decoding over 3 rounds of syndrome extraction, since we are limited by the surface code logical error rate. We simulate injecting the $\ket 0$ and $\ket +$ states and combine the error rates to get the chance of either an $X$ or $Z$ error occurring on the injected state. The results are shown in Figure~\ref{fig:hyps-injection} and plotted against the HGPS memory experiment. We find that the combined error rate of the injection gadget is well-described by $p_\text{inj} = ap^2$ with $a=3047 \pm 60$. %We expect post selection-based schemes may improve this .

\begin{figure}
    \centering
    \includegraphics[width=0.9\linewidth]{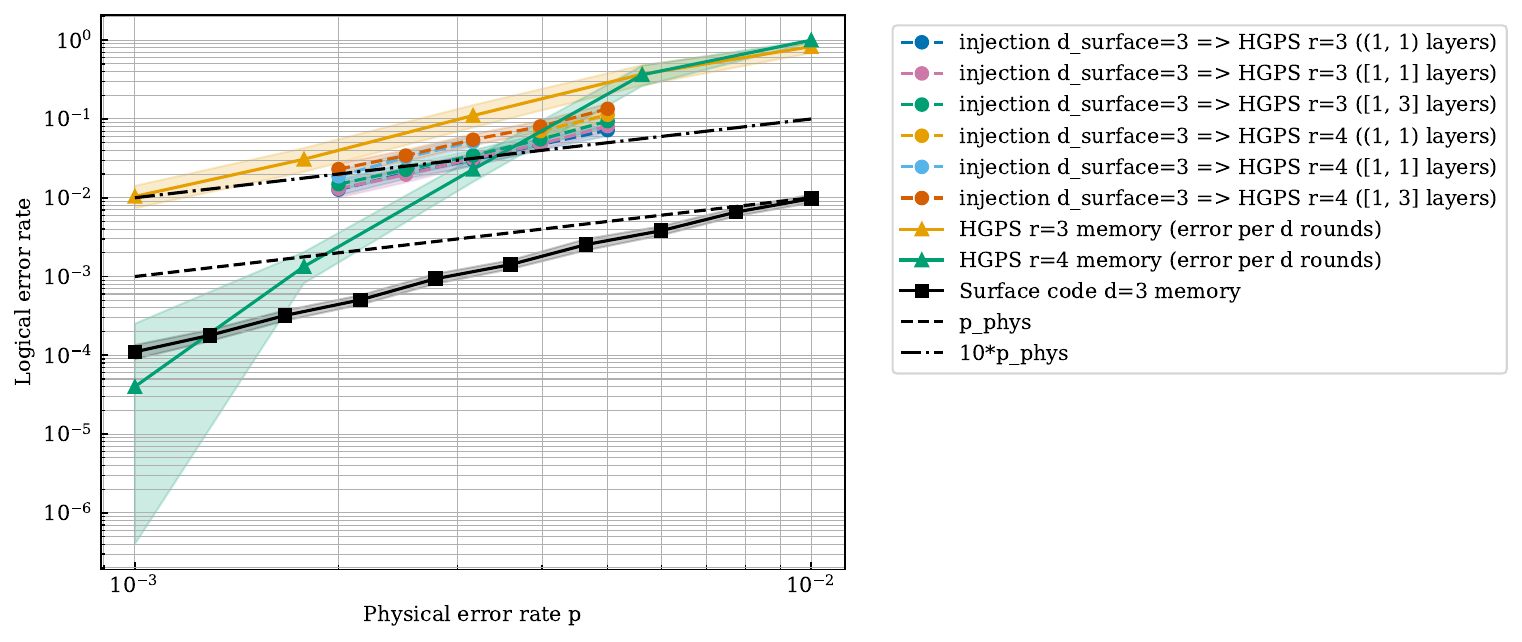}
    \caption{Logical error rate of injected magic states using universal adapters\cite{UniversalAdaptersIBM}.}
    \label{fig:hyps-injection}
\end{figure}

\subsection{Resource estimation}
\subsubsection{GHZ State}
Firstly, we compare the CQLU optimized $\ket{GHZ}$ preparation given in Figure~\ref{fig:state-prep} and with transversal surface code baseline in Figure~\ref{fig:ghz} . Benefiting from virtual operations and high rate, $GHZ$ states can be prepared using CQLU instructions with more than $7\times$ reduction in qubit footprint, which translates to similar savings in terms of space time volume.

\begin{figure}
    \centering
    \includegraphics[width=0.95\linewidth]{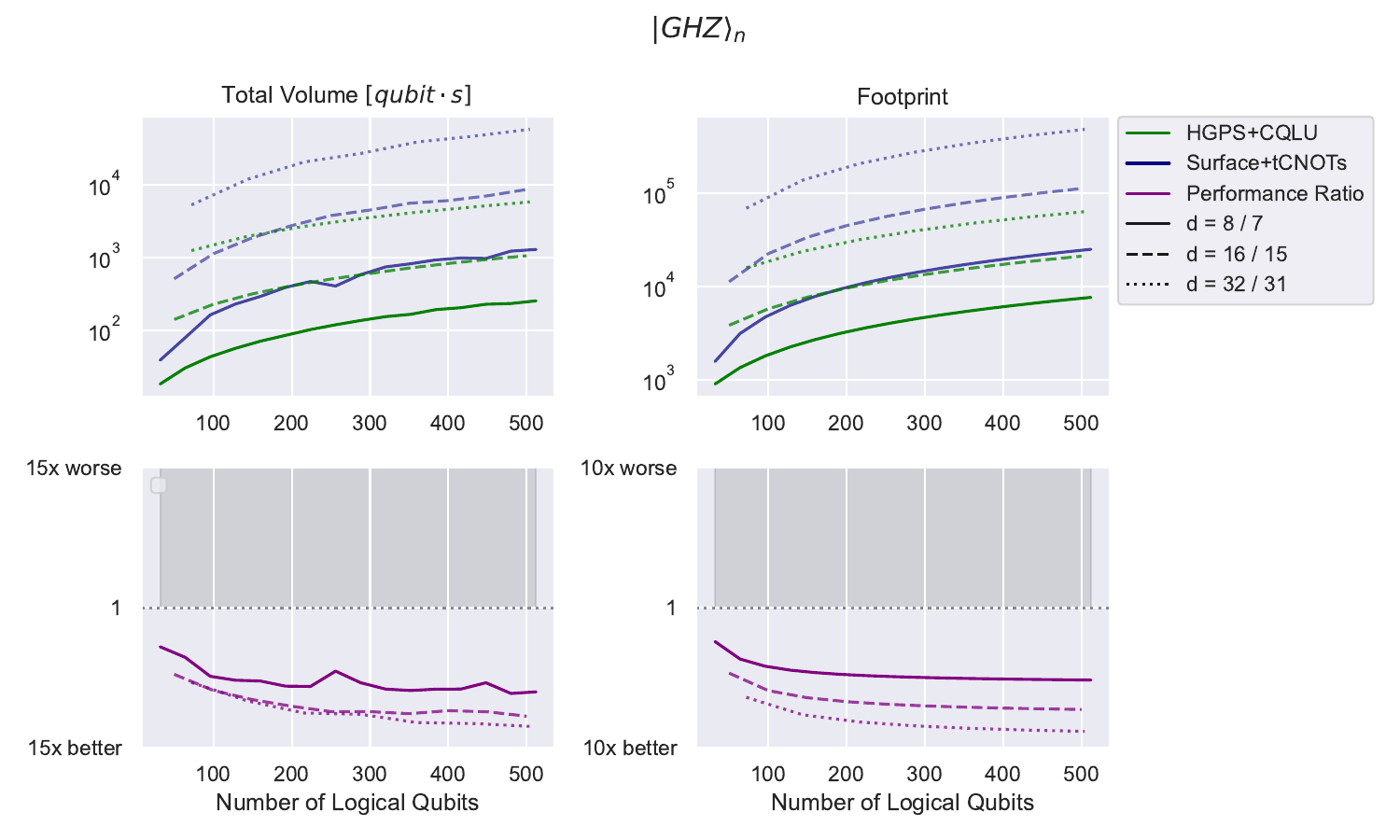}
    \caption{Space-time costs of producing a $\ket{GHZ}$.}
    \label{fig:ghz}
\end{figure}

\subsubsection{Magic State Distillation}
We use a one- or two-level magic state distillation protocol using codes given in  \cite{hastingshaah} and on HGPS with $r = 3,4,5,6$ and concatenated HGPS with surface codes, and on surface codes with distances up to $45$. We use a universal adapter to inject noisy magic states from small surface codes (see Figure~\ref{fig:hyps-injection}) into the HGPS code. The per-state resource estimations are presented in Figure~\ref{fig:msd}. Overall, we observe up to 2x footprint reduction, while the space-time volume can be 3x worse. We lose some footprint advantage due to the need for additional $\ket{i}$ states for $X/Y$ reactive measurements, and the overall volume is worse, likely due to the low fidelity magic state injection protocol, requires us to use higher distance MSD protocols or multi-level protocols earlier than needed compared to surface code baselines. 

While magic state cultivation \cite{cultivation} can be used for both RASCqL MSDs and the surface code baseline, RASCqL requires a universal adapter construction to make use of surface code cultivation, which adds substantial overhead and results in the lack of a clear benefit for RASCqL, as seen in Figure~\ref{fig:msd-inj}. We expect the MSD volume may improve further by adapting similar post-selection techniques directly to qLDPC codes, but we leave this investigation for future work.

\begin{figure}
    \centering
    \includegraphics[width=0.95\linewidth]{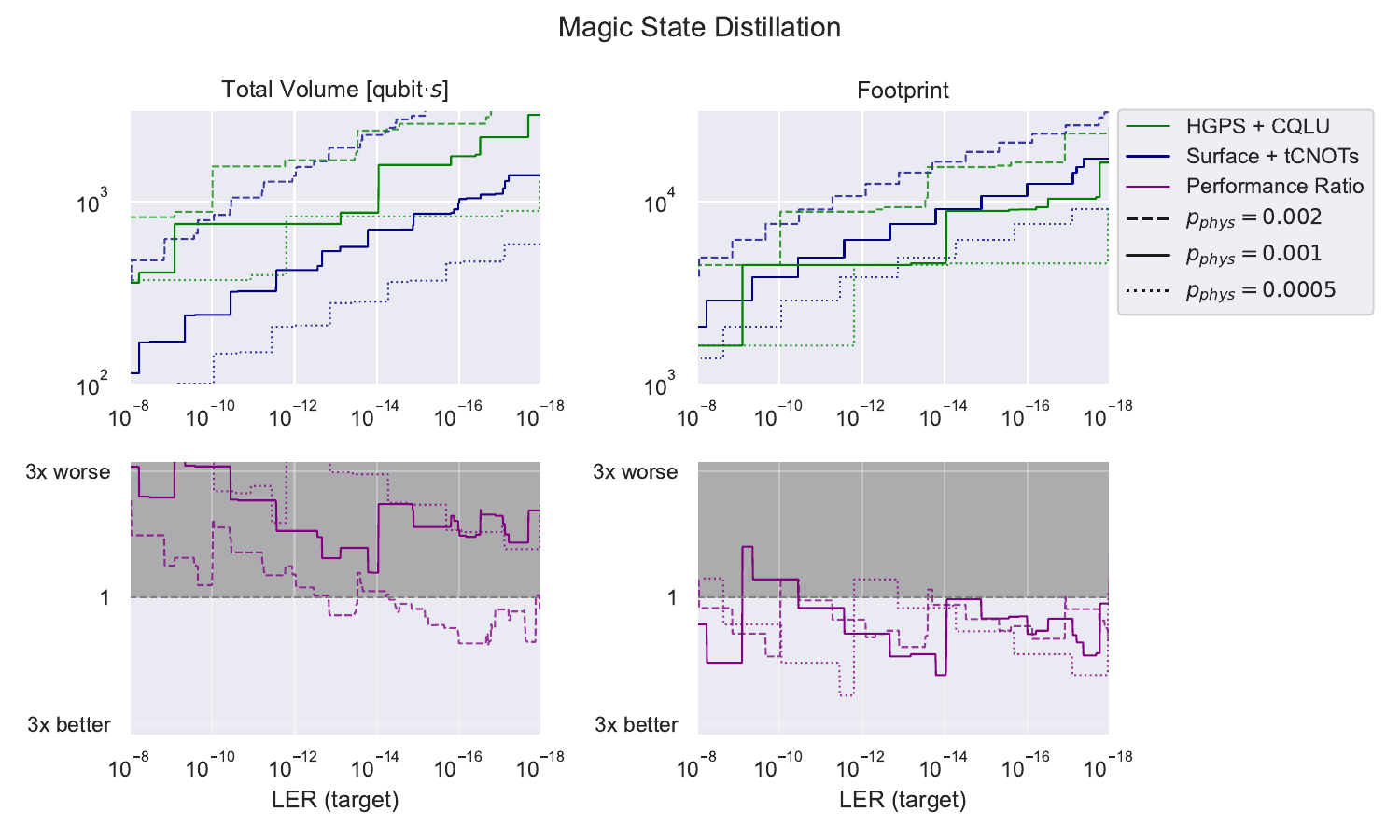}
    \caption{Space-time costs of producing high fidelity magic states in block.}
    \label{fig:msd}
\end{figure}
\begin{figure}
    \centering
    \includegraphics[width=0.95\linewidth]{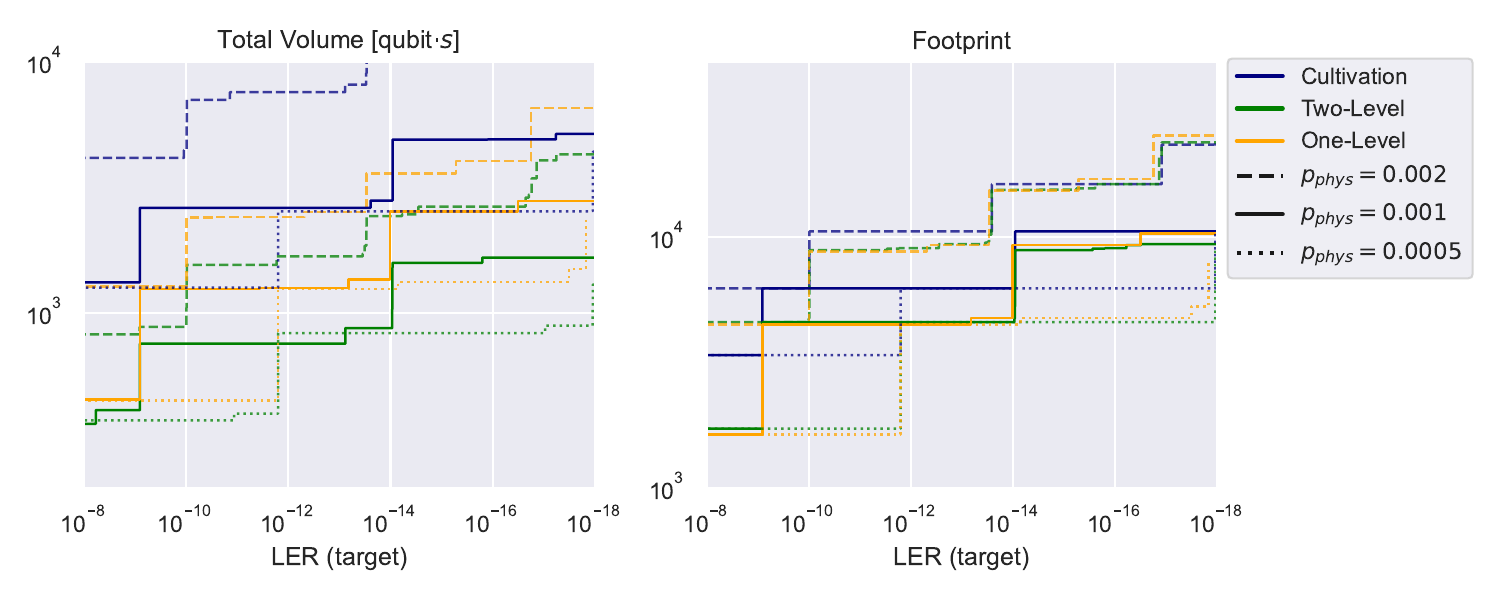}
    \caption{Comparison between concatenation schemes for MSD protocols.}
    \label{fig:msd-inj}
\end{figure}

\subsubsection{Adders}
Finally, we compare implementations of adders on CQLUs against the surface code baseline in Figure~\ref{fig:adder-cost}. While each $\mathcal{Q}_4$ can implement $8$ MAJ blocks in parallel, the optimized surface code adder require $9$ patches each (including 6 patches for $CZ$ corrections that must be kept alive for sequential reactive measurements), leading to up to $7.84\times$ difference in footprint, as reflected in the footprint comparison. For space-time volume, we evaluate the RASCqL adder with reaction-time $t_r=3.9~ms$, the same as the logical cycle time, resulting in up to $1.25\times$ reduction in Clifford volume. The Clifford volume is only $50\%$ worse even with a reaction time $10\times$ that of the surface code. Note that here we do not account for the cost to generate $\ket{T}$ states, since the scheduling and mapping of magic state production likely play a large role in the resulting resource estimation, which we leave for future work.

\begin{figure}
    \centering
    \includegraphics[width=0.95\linewidth]{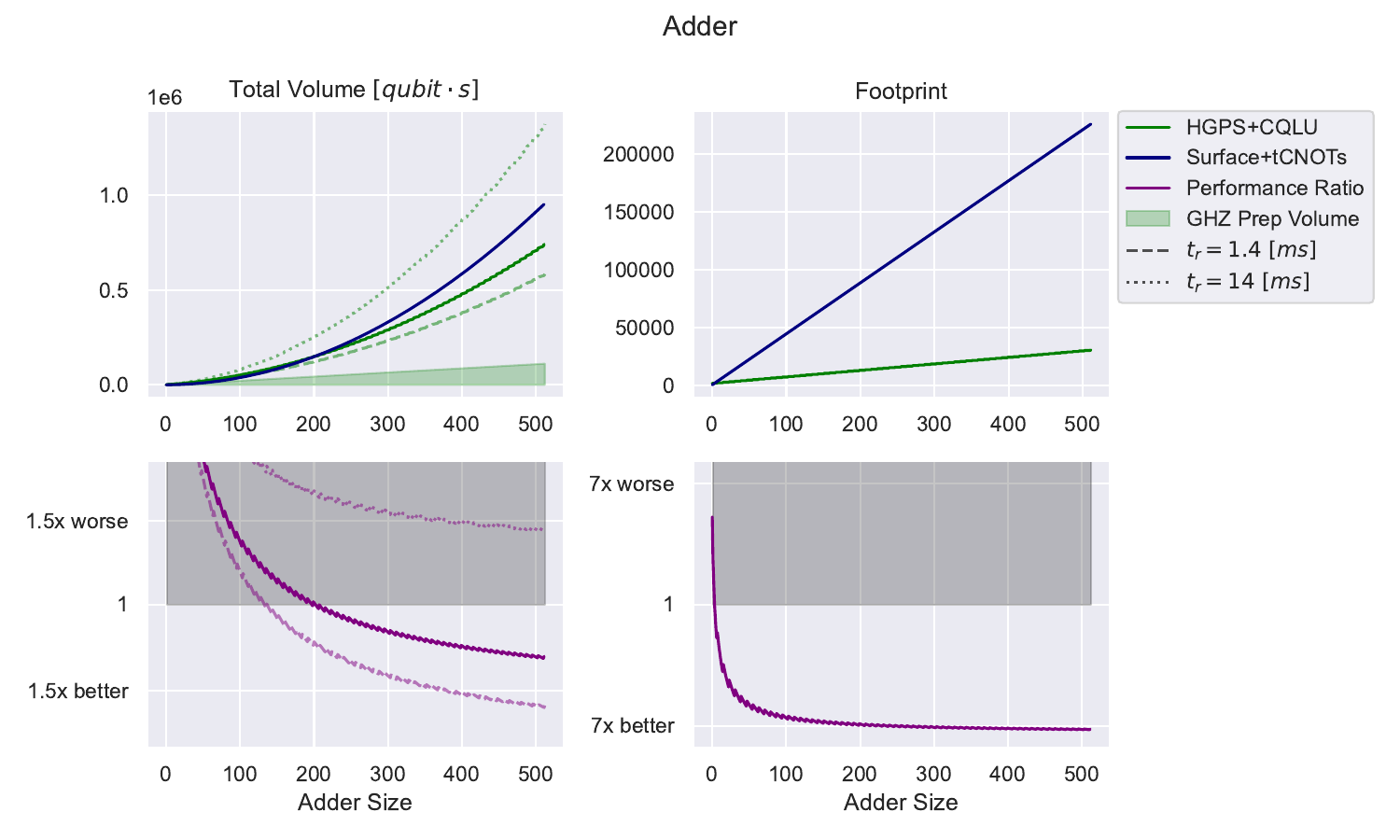}
    \caption{Space-time costs of adders using $HGPS_4$ and Surface code with $d= 7$. The dashed lines represent the resource estimates for $t_r = 1.4~[ms]$, the baseline cycle time, to $t_r = 14~[ms]$, $10\times$  baseline. Cost of magic states are not included in this analysis.}
    \label{fig:adder-cost}
\end{figure}

\section{Discussions}\label{sec:conclusion}

In this paper, we proposed an architecture based on qLDPC codes with a logical ISA that implements specific functional subroutines such as state preparation and quantum arithmetic directly in block, with the potential to outperform surface code based architectures in qubit-footprint \emph{and space-time volume}. In particular, we show an implementation of adders with reaction time limited operations that achieves up to a 7x reduction in footprint while also achieving space-time reduction in Clifford volume, and state factories with comparable costs to SoTA surface code baselines. This is achieved by carefully co-designing between circuit compilations at the functional layer, code and ISA constructions at the logical layer, and RNAA realizations at the physical layer.

While our architecture has the capability to support subroutines used in practical algorithms, there remain several important gaps to scale our architecture to end-to-end resource estimation. 

First, our reaction-time-limited model assumes efficient provisioning, storage, and routing of resource states, which may affect real end-to-end performance. A full-system scheduler that jointly optimizes footprint, latency, and state throughput is necessary to close this gap. Second, decoder latency plays a critical role. Although proposals for sub-millisecond qLDPC decoders exist \cite{muller2025improvedbeliefpropagationsufficient}, we conservatively assumed decoding time on the order of the QEC cycle ($3.9~ms$). A detailed sensitivity analysis across decoder latency, movement time, gate duration, and measurement overhead will be essential for realistic architectural comparisons.

Finally, while we focused on HGPS codes due to favorable RNAA embeddings and intrinsic symmetries, our code-modification framework applies broadly to other QECC families. A systematic search for compute-efficient qLDPC codes, co-optimizing automorphism structure, LDPC constraints, and hardware compatibility, may uncover architectures with improved flexibility across workloads and platforms.

Overall, RASCqL suggests a shift in perspective: rather than replicating a universal RISC-style ISA in qLDPC codes, or using them as quantum memory, we can instead treat them as specialized accelerators for dominant algorithmic subroutines. RASCqL establishes the question of how far this CISQ philosophy can scale toward full utility-level workloads.

\section{Acknowledgments}

We thank Yifan Hong for early discussions on Simplex codes that inspired the code modification constructions; Harry Zhou and Casey Duckering for helpful discussion on transversal architecture with reconfigurable neutral atom arrays; and Chris Kang for helpful discussions on reaction-time limited compilation. We also thank Chris Kang and Yifan Hong for their helpful feedback on the manuscript.

WY acknowledges support by the NSF GRFP Program under Grant No. 2023360346. FTC is the Chief Scientist for Quantum Software at Infleqtion. This work is funded in part by the STAQ project under award NSF Phy-232580; in part by the US Department of Energy Office of Advanced Scientific Computing Research, Accelerated Research for Quantum Computing Program; and in part by the NSF Quantum Leap Challenge Institute for Hybrid Quantum Architectures and Networks (NSF Award 2016136), in part by the NSF National Virtual Quantum Laboratory program, in part based upon work supported by the U.S. Department of Energy, Office of Science, National Quantum Information Science Research Centers, and in part by the Army Research Office under Grant Number W911NF-23-1-0077. The views and conclusions contained in this document are those of the authors and should not be interpreted as representing the official policies, either expressed or implied, of the U.S. Government, or the National Science Foundation. The U.S. Government is authorized to reproduce and distribute reprints for Government purposes notwithstanding any copyright notation herein.

%%%%%%%%% -- BIB STYLE AND FILE -- %%%%%%%%
\bibliographystyle{unsrt}
\bibliography{refs}
%%%%%%%%%%%%%%%%%%%%%%%%%%%%%%%%%%%%

\clearpage
\appendix\label{appx:proofs}

\section{Preliminary}
In this section, we write down the notations used for quantum codes and logical operations, as well as definition for the Simplex code and its quantum product variants. We also survey some useful facts and propositions from prior work as they may aid understanding of main results.

\subsection{General Notations}
\begin{definition}[Binary vector space]
    Let $\mathbb{F}_2^n$ denote an $n$ dimensional binary vector space. Binary matrices $M\in \mathbb{F}_2^{r\times c}$, with $r$ rows and $c$ columns, are used to represent linear operators or a subspace defined by its rows.
    Furthermore, let $\rows{M} = \{w_1,..,w_r\}$ and $\cols{M} = \{v_1,...,v_n\}$ refer to row- and column- vectors of $M$, we use $\rs{M} = \rs{\rows{M}}$, and $\ns{M} = \{x: Mx = 0\} \subset \mathbb F_2^n$ to denote the rowspace and nullspace of $M$, which we may refer to as $M, M^\perp$ with a slight abuse of notations. 
\end{definition}

\begin{definition}[Hilbert space]
    The state vector of a qubit can be described by a unit vector on the complex plane $\mathbb C^2$. The state space of $n$ qubits is described by the $2^n$ dimensional Hilbert space $\mathbb{H}^n = (\mathbb C^2)^{\otimes n}$. We refer to $\mathbb H^n$ as a $n$-qubit Hilbert space, and a $2^k$ dimensional subspace as a $k$-qubit subspace.
\end{definition}

We can construct new matrices by stacking $M' = \begin{bmatrix}M_1&M_2\\M_3&M_4\end{bmatrix}$. Following conventions in literature, a delineator may be added when stacking two matrices $M' = [M_1|M_2]$. We use $\otimes$ to denote Kronecker product, and $\oplus$ to denote direct sum (as opposed to the Kronecker sum): $M_1\oplus M_2 = \begin{bmatrix}M_1&\\&M_2\end{bmatrix}.$ Finally, we say $M$ is $w$-bounded if $ \wt{c}\leq w\forall c\in \rows{M}$.

\begin{definition}[Binary Linear Code] 
A $[n,k,d]$ binary linear code $\mathcal{C}$ is a $k$-dimensional subspace of $\mathbb{F}_2^{n}$ s.t. $\min_{c\in \mathcal{C}} \wt{c} \geq d$. $\mathcal{C}$ is specified by:
    \begin{enumerate}
    \item a generator matrix $G\in\mathbb{F}_2^{n_k\times n}$ where $\mathcal{C} = \rs{G}$, the row-space of $G$. Rows of $G$ can be interpreted as a basis for the codewords in $\mathcal{C}$.
    \item a check matrix $H\in\mathbb{F}_2^{n_c\times n}$ where $\mathcal{C} = \ns{H}$. Rows of $H$ can be interpreted as parity checks on the codewords.
    \end{enumerate}
    We use $\mathcal C(G,H)$ to denote a classical code with generator matrix $G$ and check matrix $H$.
\end{definition}

As constraining to each independent parity check being $0$ removes a degree of freedom from the vector space, we have $\dim(G) = n-\dim(H), \dim(H) = n-k$. Since $\rows{G}$ spans the null space of $H$, also have that $HG^T = 0$, and $\rs G\cap \rs H = e_0$. Finally, we note that there may exist many choices of $G$ and $H$ with $n_k\geq k$ and $n_c\geq n-k$. Choices of $H,G$ satisfying specific criterion will become a motif in following studies.

\begin{definition}[CSS Code]
A $[[n,k,d]]$ quantum CSS code $\mathcal{Q}$ is a $k$-qubit subspace of $\mathbb H^n$ such that any two distinct codewords differ by a Pauli operator of weight at least $d$.
% That is, whenever $\kpsi{1}, \kpsi{2}\in \mathcal{Q}$, $\kpsi{1} = P\kpsi{2}$ only if $P$ is a Pauli operator with weight at least $d$. 
$\mathcal{Q}$ contains states stabilized by a set of Pauli operators $S  = S_X\sqcup S_Z$; that is, $\mathcal{Q} = \{\kpsi{} | P\kpsi{} = \kpsi{}\forall P\in S\}$.
%\footnote{We only consider CSS codes here where the stablizers may contain only $X$ or only $Z$.} 
More compactly, $\mathcal{Q}$ can be specified by:
    \begin{enumerate}
        \item two classical binary linear codes $\mathcal{C}_X, \mathcal{C}_Z\subset \mathbb F_2^n$ with $\mathcal{C}_Z^\perp \in \mathcal{C}_X$, where $\mathcal{C}_X,\mathcal{C}_Z$ defines the $X,Z$ observables of $\mathcal{Q}$.
        \item two check matrices $H_X\in \mathbb F_2^{n_X\times n},  H_Z\in \mathbb F_2^{n_Z\times n}$ with $H_XH_Z^T = 0$, where $H_X, H_Z$ defines the $S_X$, $S_Z$ stabilizers. 
    \end{enumerate}
\end{definition}

%A convenient way to specify both quantum and classical codes in one definition is to use chain complexes:
%\todo{This is useful if we want to write homomorphic CNOT proofs (e.g. transversal teleportation between SHYPS and HGPS). Remove following definitions if it doesn't end up making its way in the manuscript.}
%\begin{definition}[Chain complexes]
%   A length-$l$ chain complex $\mathcal{V}$:
%    \begin{equation*}
%        V_l \xrightarrow{\partial_l}V_{l-1} \xrightarrow{\partial_{l-1}}\cdots \xrightarrow{\partial_1}V_0
%    \end{equation*}
%    is a sequence of vector spaces $\{V_i\}_{i=0}^l$ and boundary maps $\{\partial_i :V_i\to V_{i-1}\}_{i=1}^l$ with the condition $\partial_{i-1}\partial_{i} = 0$.
%    \begin{itemize}
%        \item Given a check matrix $H$, classical code $\mathcal{C}$ is a length-1 chain complex $$C_1\xrightarrow{H}C_0,$$ where $C_0= [n_c], C_1=[n]$ corresponds to the parity checks and data bits.
%        \item Given check matrices $H_X, H_Z$, quantum code $\mathcal{Q}$ is a length-2 chain complex $$Q_Z\xrightarrow{H_Z^T}Q_d\xrightarrow{H_X}Q_X,$$ where $Q_X = [n_X], Q_d = [n], Q_Z = [n_Z]$ corresponds to $X$ stabilizer checks, data qubits, and $Z$ stabilizer checks.
%    \end{itemize}
%\end{definition}

\subsection{Hypergraph Product Codes}
A simple way to construct quantum codes is by taking tensor-products of two classical codes with the identity. We obtain a quantum CSS code as a result which we call a hypergraph product code, or HGP for short:
\begin{definition}[Hypergraph Product (HGP) Code]
Given two classical codes $\{\mathcal{C}_i\}_{i = 0,1}$ with $H_i$, parameters $[n_i,k_i,d_i]$, and generators $G_i$ as seeds, and let $\mathcal{C}_i'$ be the code obtained with checks $H_i^T$, parameters $[m_i,k'_i,d'_i]$, and generators $G'_i$, we can construct a quantum hypergraph product code $\mathcal{Q}_{HGP}$ with stabilizer checks:
\begin{align}
    H_X &= [H_1\otimes I_{n_2}|I_{m_1}\otimes H_2^T],\\
    H_Z &= [I_{n_1}\otimes H_2| H_1^T\otimes I_{m_2}].
\end{align}
Furthermore, $\mathcal{Q}_{HGP}$ has parameters
\begin{align}
    n &= n_1n_2+m_1m_2\\
    k &= k_1k_2 + k'_1k'_2\\
    d &= min(d_1,d_2,d'_1,d'_2),
\end{align}
    and canonical basis of logical operators \cite{partition}:
\begin{align}
\\
    G_X &= [P(H_1)\otimes G_2 | 0^{m_1m_2\times m_1n_2}] \oplus [0^{n_1m_2\times n_1n_2}| G'_1\otimes P(H^T_2) ]\\
    G_Z &= [G_1\otimes P(H_2) | 0^{m_1m_2\times m_1n_2}] \oplus [0^{n_1m_2\times n_1n_2}| P(H^T_1)\otimes G'_2 ],
\end{align} where $G_i, G'_i$ is in the reduced-row-echelon form, and $P(H)$ is constructed using the pivots of $G$ such that $PG^T = I$.
\end{definition}

\subsection{Logical Operations}
In the broadest definition, a logical operation refers a transformation of the data (qu)bits that preserves the code. Here we are interested in logical operations that are both non-trivial and distance preserving. If we assume qubit permutations are noiseless (which is especially the case when a virtual relabeling of qubit position suffices), the resulting logical operations are trivially distance preserving. We call these operations code automorphisms:
\begin{definition}[Automorphisms in Classical Codes]\label{fact:auto}
    Given a code $\mathcal{C}(G,H)$, $\sigma\in\mathbb P_{k}$ is a code automorphism if $\exists g\in\gln{k}$ s.t. 
    \begin{equation}
        G\sigma = gG.
    \end{equation}
    $g$, referred to as the automorphism gate, describe the logical action of the code automorphism under the basis defined by $G$. 
\end{definition}

There are some useful facts we can show about code automorphisms. These facts follow from several propositions in \cite{berthusen2025automorphismgadgetshomologicalproduct} but are also to establish intuition.  We survey some of these facts here in our notations.
\begin{fact}[Automorphism group]
    Given a code $\mathcal{C}(G,H)$ and automorphisms $\Sigma$ that preserves $\mathcal C$ (i.e. right action by $\sigma\in\Sigma$ is an isomorphism from $\mathcal{C}$ onto itself); $\Sigma$ is a group under multiplication. We say $\mathcal{G}_\Sigma$ is the automorphism group of $\mathcal{C}$, or $\aut{\mathcal{C}} = \mathcal{G}_\Sigma$, if 
    \begin{equation}
        \mathcal{G}_\Sigma \cong \{\sigma: \exists g\in  \gln{k}\ s.t.\ G\sigma = gG \}.
    \end{equation}
\end{fact}

\begin{proof}
    Clearly the identity preserves $\mathcal{C}$, composing two isomorphisms lead to another isomorphism $\mathcal{C}$, and compositions are associative as matrix multiplication is associative. For any $g, \sigma: G\sigma =gG$, we have that 
    \begin{equation}
        G\sigma =gG \implies g^{-1}gG\sigma = g^{-1}G\sigma^2 \implies G\sigma^T = g^{-1}G,
    \end{equation}
    hence $\sigma^T$ preserves $\mathcal{C}$, with $\sigma\sigma^T = I$.
\end{proof}

\begin{fact}[Automorphism Gate]
    Given a code $\mathcal{C}(G,H)$. Let $\mathcal{L}_G(\mathcal{C})=\{g:\exists \sigma \st gG = G\sigma\}$ be the set of automorphism gates. Then, $\mathcal{L}_G(\mathcal{C})$ is a group, $\mathcal{L}_G(\mathcal{C})\cong \aut{\mathcal C}$ if columns of $G$ is unique.
\end{fact}
\begin{proof}
    It's simple to show that $\mathcal{L}_G(\mathcal{C})$ is a group.

    We can show that the mapping $\phi_G: \aut{\mathcal{C}}\to \mathcal{L}_G$ s.t. $\phi_G(\sigma) = g\iff gG = G\sigma$ is an group homomorphism. Indeed, $\sigma_1G = g_1G, \sigma_2G= g_2G\implies\sigma_1\sigma_2G = g_1g_2G \implies \phi_G(\sigma_1\sigma_2) = g_1g_2$. 
    
    When columns of $G$ are unique, $\ker{\phi_G} = I$ since any non-trivial permutation lead to non-trivial logical actions; $\phi_G$ is in fact an isomorphism.
\end{proof}

Unless otherwise specified, we will assume columns of $G$ are unique, in which case the notation $\aut{\cdot}$ is used synonymously to refer to the automorphism group or the group of resulting logical actions. In cases where non-unique columns are important for distance guarantees, we can keep track of how many times a particular column is duplicated before removing them, perform the code modification, and then duplicate the same column after modification. 

\begin{fact}[Automorphism of Checks]
    Given $\mathcal{C}(G,H)$ and $\sigma\in \aut{\mathcal{C}}$, we also have that
    \begin{equation}
       \exists h\in\gln{k} \st H\sigma = hH.
    \end{equation}  
     Furthermore, we say $\sigma$ is a matrix automorphism if $h\in\mathbb P_k$.
\end{fact}
This follows from the fact that $\mathcal{C} = \ker H$, and $\sigma(c) \in \mathcal{C}\forall c\in \mathcal{C} \implies H(\sigma(c))^T =  (H\sigma^T) c^T = 0 \forall c\in \mathcal{C} \implies \rs H = \rs {H\sigma^T}$.
\subsection{Simplex Codes and Quantum Variants}
A classical code with expansive automorphism group is the simplex code:
\begin{definition}[Simplex Codes]\label{def:simplex}
    An $r$-dimensional binary simplex code $\mathcal{S}_r$ is a $[2^r-1, r, 2^{r-1}]$ classical linear code with generator matrix $G_r$ s.t.
    \begin{equation}
        \cols{G_r} = \{b_1,...,b_{2^r-1}\}
    \end{equation}
    where elements of $b_i$ correspond to digits of $i$ in binary. Unless otherwise specified, let the check matrix be given as cyclic matrices in \ref{tab:simplex}. 
\end{definition}
\begin{table}
    \centering
    \begin{tabular}{c|cccc}
         $r$&  3&  4&  5&  6\\\hline
         $H_r$&  $I+x+x^3$&  $I+x+x^4$&  $I+x^2+x^5$&  $I+x+x^6$\\
    \end{tabular}
    \caption{Check matrices for the simplex code, where $x$ denote the cyclic permutation matrix of dimension $2^r-1$. }
    \label{tab:simplex}
\end{table}

\begin{definition}[Hypergraph Product of Simplex (HGPS) Codes]\label{def:hgps}
    Let $\mathcal S_r, H_r$ be defined as \ref{def:simplex}. Define $\mathcal{Q}_{HGPS, r}$ as the hypergraph product code with seeds $H, H^T$:
    \begin{align}
        H_X & = [H\otimes I | I\otimes H^T]\\
        H_Z & = [I\otimes H | H^T\otimes I].
    \end{align}
    $\mathcal{Q}_{HGPS, r}$ has parameters $[[2(2^r-1)^2, 2r^2, 2^{r-1}]]$ and generators 
    \begin{align}
        G_X &= [P\otimes G]\oplus [G\otimes P],\\
        G_Z &= [G\otimes P] \oplus  [P\otimes G].\\
    \end{align}
    Note the logical operators naturally partitions into two sets with disjoint support on the left and right sector. 
\end{definition}

\subsection{Automorphism group of the Simplex codes}
It is well known that the automorphism group of the Simplex code contains all linear transformations.  This in fact saturates the upper-bound on the possible automorphism group size in any codes with the same number of logical bits, as a result of Corr 3.3 in ~\cite{berthusen2025automorphismgadgetshomologicalproduct}. 
\begin{fact}[Simplex Automorphism] 
Let $\mathcal{S}_r$ be defined as in Def~\ref{def:simplex}. Then, $\mathcal{L}_G(\mathcal{S}_r) = \gln{r}$.
\end{fact}

While HGPS do inherit all logical gates corresponding to the product of Simplex automorphisms by construction, the physical transformations required remains physical permutations (an hence fault tolerant) only if the classical automorphism is also a matrix automorphism on the checks, which is not true in general \cite{berthusen2025automorphismgadgetshomologicalproduct}.

One approach to get around this, as proposed in \cite{berthusen2025automorphismgadgetshomologicalproduct}, is to gauge out all problematic logical qubits:
\begin{proposition}[Left Sector HGPS Automorphism]\label{prop:ls-hgps}
    From $\mathcal{Q}_{HGPS,r}$ as defined as in \ref{def:hgps}, we can obtain a subsystem code $\mathcal{Q}^L_{HGPS,r}$, the left-sector HGPS, where logical qubits supported on the right sector is designated as gauge qubits. We have:
    \begin{equation}
        \gln{r}\times\gln{r}\leq \mathcal{L}_G({\mathcal{Q}^L_{HGPS,r}}).
    \end{equation}
    i.e., for all $g = g_1\otimes g_2\in \gln{r}\times \gln{r}$, $\exists \sigma = \sigma_1\otimes \sigma_2, h = h_1\otimes h_2$, such that the action of $\sigma\oplus h$ on the physical qubits applies the logical transformation
    \begin{align*}
        G_X^L\sigma &= gG_X^L\\
        G_Z^L\sigma &= g^{-T}G_Z^L
    \end{align*}
    fault-tolerantly to the logical qubits supported on the left sector. 
\end{proposition}
The proof follows from [Theorem~5.5, \citealp{hgp}] and [Lemma~VIII.12, \citealp{ComputingQLDPC}], 
which we survey here.

\begin{proof}
    Given $g_1, g_2$, we can choose 
    \begin{align*}
        \sigma_1 &: g_1^{-T}G = G\sigma_1,\\
        \sigma_2 &: g_2G = G\sigma_2,\\
        h_1 &: h_1H = H\sigma_1,\\
        h_2 &: h_2^{-T}H = H\sigma_2.
    \end{align*}
    Indeed, we can check that the stabilizers are preserved:
    \begin{align*}
        H_X\sigma\oplus h 
        &= [H\sigma_1\otimes \sigma_2 | h_1\otimes H^T h_2] \\
        &= [H\sigma_1\otimes \sigma_2 | h_1\otimes (h_2^{T}H)^T]\\
        &= [h_1H\otimes \sigma_2| h_1\otimes \sigma_2H^T]\\
        &= h_1\otimes \sigma_2 H_X\\
        H_Z\sigma\oplus h^{-T} 
        &= [\sigma_1\otimes H\sigma_2 | H^Th_1^{-T}\otimes h_2^{-T}] \\
        &= [\sigma_1\otimes H\sigma_2 | (h_1^{-1}H)^{T}\otimes h_2^{-T}] \\
        &= [\sigma_1\otimes h_2^{-T}H| \sigma_1H^T\otimes  h_2^{-T}]\\
        &= \sigma_1\otimes h_2^{-T} H_Z.
    \end{align*} 
    Then, to show that $\sigma$ leads to the desired logical action on the left sector, we first note that
%    \begin{align}
%        I 
%        &= L_X\sigma \sigma^TL_Z^T \\
%        &= G\sigma_1\sigma_1^TP^T\otimes P\sigma_2^T\sigma_2G^T\\
%        &= h_1G\sigma_1^TP^T\otimes P\sigma_2G^Th_2^{-T}
%    \end{align}
%    which implies 
    \begin{align}
        I 
        &=P\sigma_1\sigma_1^TG^T = P\sigma_1G^Tg_1^{-1} \\
        &\implies  P\sigma_1G^T = g_1 = g_1PG^T \\
        &\implies (P\sigma_1 - g_1P)G^T = 0 \\
        &\implies P\sigma_1 = g_1P +M
    \end{align}
    for some $M\in \ker(G) = \rs H\implies M = aH$ for some matrix $a$. Then, it follows that 
    \begin{align}
        G_X^L \sigma 
        &= P\sigma_1\otimes G\sigma_2  \\
        &= g_1P\otimes g_2G +  aH\otimes G \\
        &= gG_X^L + a' (H_X +H_X^R),  
    \end{align}
    where $a' = a\otimes G$ and $H_X^R$ correspond to a transformation on right sector qubit, which does not effect the left sector logical operators. We can similarly show that $G_Z^L \sigma\simeq g^{-T}G_Z^L$ up to stabilizers and right-sector operations. 

    Finally, these automorphism gates are fault-tolerant. While implementations of $h$ require entangling gates and decrease the distance in general, it only effects qubits supported on the right sector, which we preemptively designated as gauge qubits. With, all logical observables being disjoint from physical qubits where entangling gates are applied, the distance of the dressed operators cannot decrease.
\end{proof}

With these operators, $\mathcal{Q}_{HGPS,r}^L$ can implement arbitrary CNOT-type operators efficiently using the synthesis routine shown in \cite{ComputingQLDPC} with doubled qubit overhead. This approach unfortunately cannot implement diagonal gates. 

\section{Code Modifications}
Now, we have a language to describe the gap. Firstly, while having rich classical automorphisms are helpful for quantum computation, it also implies a poor encoding rate. This is often unnecessary as much smaller subset of gates usually suffice for any \emph{practical} computation, leading to a natural question of whether we can design codes with better or constant encoding rates that contain exactly the gates we need. We investigate this question in section~\ref{ssec:prescription}, and show precise bounds on the number of physical qubits that suffices to augment automorphisms to a given code. Secondly, access to classical automorphisms isn't helpful unless we can lift them to quantum codes, and existing approaches either sacrifices LDPC property for a specific class of code or throw away up to half of the logical qubits. It is hence another important question to study resource efficient methods to lift specifc classical automorphisms. We explore this question in section~\ref{ssec:modification}. Specifically, we find minimal modifications of the checks such that automorphism are converted to matrix automorphisms that can be implemented virtually.

\subsection{Prescribing Automorphisms}\label{ssec:prescription}
The equivalent definition for the group of automorphism gates using group action is as follows:
\begin{proposition}[Automorphism Condition]\label{prop:aut-cond}
    Let us be given a classical linear code $\mathcal{C}(G,H)$ and $\mathcal{G} = \{g_1,...,g_m\} \leq \gln{k}$. Then, $ \mathcal{G}\leq \mathcal{L}_G(\mathcal{C})$ iff
    \begin{equation}\label{eq:auto-condition}
        \orbg{c} = \{g\cdot c\ | g\in \mathcal{G}\} \subseteq \cols{G}\ \ \forall c\in \cols{G}.
    \end{equation}
     i.e. the columns of $G$ is $\mathcal{G}$-invariant under left action.
\end{proposition}
This is simply a re-statement Definition~\ref{fact:auto}, since $\cols{gG} = \cols{G\sigma} =\cols{G}\iff g\in \mathcal{L}_G(\mathcal{C})$. This definition gives a general way to construct codes with a prescribed automorphism group:
\begin{theorem}[Automorphism Completion]\label{thm:auto-expansion}
    Let us be given a $[n,k,d]$ code $\mathcal{C}(G,H)$ and $\mathcal{G} = \{g_1,...,g_m\} \leq \gln{k}$. There exists a family of $[n'\leq nm,k,d'\leq dm]$ codes $\mathcal{C}_{\mathcal{G}}(G_{\mathcal{G}},H_{\mathcal{G}})$ such that $\mathcal{C}\cong\mathcal{C}_\mathcal{G}$, $\mathcal{G}\leq \mathcal{L}_{G_\mathcal{G}}({\mathcal{C}_\mathcal{G}})$. 
\end{theorem}
\begin{proof}
    Simply take $G_{\mathcal{G}} = [g_1G|\cdots|g_mG]$ to achieve $n' = nm, d' = dm$. Then, non-unique columns can be punctured while preserving the automorphism group, simultaneously reducing both code length and distance.

    To upper-bound the redundancy, we can find the smallest $\mathcal\ G$-invariant set containing $\cols{G}$: $\mathcal O = \mathcal{G}\cdot \cols{G}= \bigcup_{c\in \cols{G}}\orbg{c}$. It suffices to take $G_\mathcal{O}$ constructed by stacking all columns $c\in \mathcal{O}$. 
    
    To find $H_\mathcal{G},$ we can assume WLOG that $G = [I|M]$, and let $M_\mathcal O = G_\mathcal{O}\setminus\cols{I_k};$ this gives
    \begin{align}
       G_{\mathcal{G}} &= \begin{bmatrix}
            I& |M_\mathcal{O}
        \end{bmatrix}\\
        H_{\mathcal{G}} &= \begin{bmatrix}
            M_\mathcal{O}^T|&I
        \end{bmatrix},
    \end{align}
    
    s.t. $H_\mathcal{G}G_\mathcal{G}^T = 2M_\mathcal{O}^T = 0$. Hence, $\mathcal{C}_{\mathcal{G}}(G_\mathcal{G}, H_\mathcal{G})$ defines a valid linear code.
    
    It's also easy to check that $\cols{G_\mathcal{G}} = \mathcal{O}$ is $\mathcal{G}$-invariant under left action: for all $c\in\mathcal{O}$, $c = g_G\cdot c_G$ for some $g_G\in \mathcal{G}, c_G\in \cols{G}$, it follows $g\cdot c =  g\cdot g_G\cdot c_G\in \orbg c\in \mathcal{O}$ for all $g\in\mathcal{G}$. This implies $\mathcal{G}\leq \mathcal{L}_{G_\mathcal{G}}({C_\mathcal{G}})$ by Prop~\ref{prop:aut-cond}

\end{proof}

Note depending on $\mathcal{G}$, the size of the code grows with the size of the automorphism group we want to expand to. On one extreme, expanding the automorphism to $\gln{k}$ potentially yield an exponential increase in code length; in fact, starting with any code $\mathcal{C}$, the $\mathcal{C}_{\gln{k}}$ we obtain (after removing redundancy) is exactly the simplex code $\mathcal{S}_r$, with $n = 2^{k}-1$. On the other hand, if we fix $\mathcal{G}$ to be small, e.g. $\mathcal{G} = \{I,g_{CNOT}\}$ correspond to a single CNOT, we can obtain sharper guarantees on $n',d'$.

%\begin{corollary}[Exansion: Communting CNOTs]
% Let us be given a code $\mathcal C(G,H)$ and $g_{CNOTs}$ correspond to a set of commuting CNOT gates. There exists code $\mathcal{C}'(G', H')$ s.t. $g_{CNOTs}\in \mathcal{L}_{G'}({\mathcal{C}'})$. Furthermore, if $\mathcal{C}$ has parameters $[n,k,d]$, $\mathcal{C_G}$ has parameters $[n',k,d']$, where $n\leq n'\leq2n$ and $d\leq d'\leq 2d$.
%\end{corollary}
%\begin{proof}
%    It follows from Theorem~\ref{thm:auto-expansion} by taking $\mathcal{G} = \rs{g_{CNOTs}} = \{I,g_{CNOTs}\}$ since a set of commuting CNOTs is an involution.
%\end{proof}

The code expansion given in Theorem~\ref{thm:auto-expansion} can maintain LDPC properties of the original check matrix if either the prescribed automorphism group is small, or if the group elements are sparse:

\begin{theorem}[LDPC]\label{thm:expansion-ldpc}
Let us be given a code $\mathcal{C}(G,H)$ and $\mathcal{G} \leq \gln{k}$. Suppose $H$  is $w$-bounded, $g$ is $t$-bounded for all $g\in \mathcal{G}$ and $|\mathcal{G}| = m$.  Then, the family of codes $\mathcal{C}_\mathcal{G}(G_\mathcal{G},H_\mathcal{G})$ given by Theorem~\ref{thm:auto-expansion} has a $(w + t+ 1)$-bounded check matrix $H^{(t)}_\mathcal{G}$ and a a $(w + m+ 1)$-bounded check matrix $H^{(m)}_\mathcal{G}$.
\end{theorem}

\begin{proof}
 Let us find $P_G,\sigma_G$ by Gaussian Elimination such that $G' = P_GG\sigma_G = [I|M]$ is in the canonical basis. Note both $H' = [M^T|I]$ and $H\sigma_G$ are valid check matrices for $G'$, hence there exists a basis change such that $P_H[M^T|I]\sigma_G^T = H$. Let $H_1 = P_HM^T$, and $H_2 = P_H$, it follows that $H_1, H_2, [H_1|H_2] = H$ are all $w$-bounded.
 
Following Theorem~\ref{thm:auto-expansion}, we can find the orbit of $M$ under $\mathcal{G}$, $M_\mathcal{O} = \begin{bmatrix} g_1\cdots g_m&M&g_1M\cdots g_mM\end{bmatrix}$ and define $G_\mathcal{G} = [I|M_\mathcal{O}]$, giving us a code $\mathcal{C_G}$ with desired automorphism group. It remains to find a $H_\mathcal{G}$ with LDPC properties. 

We begin with a valid check matrix $H^{(1)}_{\mathcal{G}}= [M_\mathcal{O}^T|I]$. First, let $\sigma_r$ be permutation that swaps row $i$ with row $i+mk+(n-k)$ for $i\leq mk$, and $P_0 = \begin{bmatrix}
    I&\\M^T&I
\end{bmatrix}$ so that $P_0\cdot \begin{bmatrix}
    g_i^T&I&\\
    M^Tg_i^T&&I
\end{bmatrix} =  \begin{bmatrix}
    g_i^T&I&\\
    &M^T&I
\end{bmatrix}$. Define $P_0' =   P_0\oplus P_0\oplus...\oplus P_0\oplus I$ and $P_a = \sigma_r^T P_0' \sigma_r$, we can do the following transformation to obtain:
\begin{align}
H^{(2)}_{\mathcal{G}}&=P_a\cdot \left[
        \begin{array}{c| c c c c c c}
          g_1^T & I & & & & & \\
          \vdots & & \ddots & & & & \\
          g_m^T & & & I & & & \\ \hline
          M^T & & & & I & & \\
          \vdots & & & & & \ddots & \\
           M^Tg^T_m & & & & & & I
        \end{array}
        \right]\\
%    &=\sigma_a^TP_a\left[
%        \begin{array}{c| c c c| c c c c}
%          M^T & & & & I & & & \\\hline
%          g_1^T & I & & & && & \\
%          M^Tg_1^T & & & & &I & & \\\hline
%          \vdots & & \ddots & & &\ddots & & \\\hline
%          g_m^T & & & I & & && \\ 
%           M^Tg^T_m & & & & && & I
%        \end{array}
%        \right]\\
%    &=\sigma_a^T\left[
%        \begin{array}{c| c c c| c c c c}
%          M^T & & & & I & & & \\\hline
%          g_1^T & I & & & && & \\
%            &M^T & & & &I & & \\\hline
%          \vdots & & \ddots & & &\ddots & & \\\hline
%          g_m^T & & & I & & && \\ 
%            & & & M^T& && & I
%        \end{array}
%        \right]\\
    &= \left[
        \begin{array}{c c c c c c c c}
          g_1^T & I & & & & && \\
          \vdots & & \ddots & & & && \\
          g_m^T & & & I & & & &\\ \hline
          M^T & & & & I & & &\\
           &M^T & & &  &I & &\\
           & & \ddots& & & &\ddots & \\
         & & & M^T& & & &I
        \end{array}
        \right]
\end{align}

Then, we can apply the basis change $P_b = I_{mk}\oplus P_H\oplus \cdots\oplus P_H$ on the bottom half of $H_\mathcal{G}^{(2)}$, and obtain
\begin{align}
H^{(3)}_\mathcal{G} &= P_bH^{(2)}_\mathcal{G}\\
        &= \left[
        \begin{array}{c c c c c c c c}
          g_1^T & I & & & & && \\
          \vdots & & \ddots & & & && \\
          g_m^T & & & I & & & &\\ \hline
          H_1 & & & & P_H & & &\\
           &H_1 & & &  &P_H & &\\
           & & \ddots& & & &\ddots & \\
         & & & H_1& & & &P_H
        \end{array}
        \right]
\end{align}

If $m\leq t$, we can construct a $(w+m+1)$-bounded check matrix by applying $g_i^{-T}$ to rows $[ik, (i+1)k-1]$, obtaining
\begin{align}
    H^{(m)}_\mathcal{G} = \left[
        \begin{array}{c c c c c c c c}
          I & g_1^{-T} & & & & && \\
          \vdots & & \ddots & & & && \\
          I & & & g_m^{-T} & & & &\\ \hline
          H_1 & & & & P_H & & &\\
           &H_1 & & &  &P_H & &\\
           & & \ddots& & & &\ddots & \\
         & & & H_1& & & &P_H
        \end{array}
        \right]
\end{align}

Finally, we construct a $(w+t+1)$-bounded check matrix. Let $P_i = \begin{bmatrix}I&\\g_{i+1}^Tg_{i}^{-T}&I\end{bmatrix}$ for $i = 1,...,k$ so that $P_i\cdot \begin{bmatrix}
    g_i^T&I&\\g_{i+1}^T&&I
\end{bmatrix} =\begin{bmatrix}
    g_i^T&I&\\&g_{i+1}^Tg_{i}^{-T}&I
\end{bmatrix} $, then apply $P_i$ to rows $[ik, (i+2)k-1]$, giving us:
\begin{align}
    H^{(t)}_\mathcal{G} = \left[
        \begin{array}{c c c c c c c c}
          g_1^T & I & & & & && \\
           & g^T_2g_1^{-T}  &I & & & & & \\
           & & \ddots & & & & & \\
           & & g_{m-1}^Tg_m^{-T}& I & & & &\\ \hline
          H_1 & & & & P_H & & &\\
           &H_1 & & &  &P_H & &\\
           & & \ddots& & & &\ddots & \\
         & & & H_1& & & &P_H
        \end{array}
        \right]
\end{align}

%\begin{align}
%    H^{(4)}_{\mathcal{G}} &= P_cH^{(3)}_\mathcal{G}\sigma_c\\
%    &= \left[
%        \begin{array}{c c c c c c c c}
%          g_1^T|0 &I|0 & & &   \\
%           & g^T_2g_1^{-T}|0  & I|0& &   \\
%           & &\ddots & &   \\
%           & && g_{m-1}^Tg_m^{-T}|0&I|0 &  \\ \hline
%          H & & & &   \\
%           &H & & &    \\
%           & &\ddots & &  \\
%         & & &H&  \\
%         & & & &H  
%        \end{array}
%        \right]
%\end{align}
Which is $(w+t+1)$-bounded, since $g_{i+1}g_{i}^{-1}\in \mathcal{G}$ is $t$-bounded. 
\end{proof}
%Clearly, $H_{\mathcal{G}}^{(4)}$ has row weight at most $t +1$ and column weights at most $w+t+1$. Finally, $H_{\mathcal{G}} =H_{\mathcal{G}}^{(4)}\times \sigma$ for some column permutation is a valid check matrix that is also $(w+t+1)$-bounded, as desired.\end{proof}

%\begin{corollary}[Local Gates]\label{cor:G-t-local}
%Let us be given a code $\mathcal C(G,H)$ and $\mathcal{G}_t  \cong \gln{t}$ denote all operations on $t$ bits. There exists $[[2n-2, k, 2d-2]]$ code $\mathcal{C}'(G', H')$ s.t. $G_{t}\in \mathcal{L}_{G'}(\mathcal{C}')$. 
%\end{corollary}

We can often find tighter bounds for specific classes of operations. We list a few here for example.
\begin{corollary}[Expansion: CNOT Fan-out]\label{cor:fan-in}
Let us be given a $[n,k,d]$ code $\mathcal C(G,H)$ and let $H$ be $w$-bounded. Let $g_{{f}}$ denote a fan-out operations on up to $w+1$ bits. There exists a $[2n,k,2d]$ code $\mathcal{C}'(G', H')$ s.t. $g_{{f}}\in \mathcal{L}_{G'}(\mathcal{C}')$, where  $H_\mathcal{G}$ is $(w+2)$-bounded.
\end{corollary}

%There are several more optimziations when the operators $g$ are specified. For example, note the generator matrix in Cor.~\ref{cor:G-t-local} has length $2n-k+t<2n$, and $P_\mathcal{G}H''_\mathcal{G}$ has weight $t$ only if there exists a row with all $1$s. This is never the case, for instance, in the case of a single CNOT, since $(I+CNOT)\cdot(b_1,b_2) = (0,b_2)$ have weight at most $1$ even though it has order $2$. Furthermore, assuming logical bit flip is free, the cost of a $g$ or $X_1gX_1$ is the same, while $g$ preserves all columns with $b_1 = 0$ and $X_1gX_1$ preserves all columns with $b_1 = 1$; at least one of them will preserve at least half of the columns, which further reduce the number of physical qubits used. It follows that the overhead for implementing a single CNOT is 

\begin{corollary}[Expansion: Local Commuting CNOTs]
Let us be given a $[n,k,d]$ code $\mathcal C(G,H)$ with $w$-bounded $H$ and a set of $m$ commuting CNOT gates $\mathcal{G} = \{g_1,...,g_m\}$ on $t$ bits. Then, there exists $[nm,k, md]$ code $\mathcal{C_G}(G',H')$ with $(w+t+1)$-bounded $H_\mathcal G$ s.t. any circuits involving only $\mathcal{G}$ can be performed virtually.
\end{corollary}

%We also note that choosing a suitable $G$ to begin with can often further reduce overhead. We can choose  the space-overhead upperbounds shown is often only tight in the case of randomly generated codes, in which case the distance upper bounds will also saturate. %We explore these trade-offs more in section~\ref{sec:apps}.

\subsection{Converting Automorphisms}\label{ssec:modification}
A code automorphism is also a matrix automorphism if we find $H\sigma = \rho H$ for some $\rho\in \mathbb P_k$. The matrix automorphism condition can be equivelently stated as follows:

\begin{proposition}[Matrix Automorphism Condition]
    Let us be given a classical linear code $\mathcal{C}(G,H)$ and $\Sigma = \{\sigma_1,...,\sigma_m\} \leq\aut{\mathcal{C}}$. Then, elements of $\Sigma$ are also matrix automorphisms iff
    \begin{equation}
        \orbs{r} = \{r\cdot \sigma\ | \sigma\in \Sigma\} \subseteq \rows{H}\ \ \forall r\in \rows{H}.
    \end{equation}
     i.e. the rows of $H$ is $\Sigma$-invariant under right action.
\end{proposition}

\cite{berthusen2025automorphismgadgetshomologicalproduct} shows that any classical linear code $\mathcal C$ admits an exponentially large check matrix where all automorphisms are also matrix automorphisms. Indeed, there exist codes such as the simplex code that saturate this bound. However, in general, the size of the check matrix can be bounded by the size of automorphisms we wish to convert to matrix automorphisms:

\begin{theorem}[Matrix Automorphism Conversion]\label{thm:auto-conversion}
    Let us be given a code $\mathcal{C}(G,H)$ and $\Sigma = \{\sigma_1,...,\sigma_m\} \leq \mathbb P_{k}$ and $\Sigma\leq\aut{\mathcal{C}}$. There exists $H_\Sigma\in \mathbb F_2^{n_r\times n}$ with $n_r\leq (n-k)\times m$ s.t. $\Sigma$ is a group of matrix automorphism on $H_\Sigma$. Furthermore, if $H$ is $w$-bounded, $H_\Sigma$ is $wm$-bounded.
\end{theorem}
\begin{proof}
    We can take $H_\Sigma = [(H\sigma_1)^T|\cdots |(H\sigma_m)^T]^T$ to obtain a check matrix with  $n_r= (n-k)\times m$, and remove non-unique rows.
    
    We can bound the number of rows removed similar to the proof of Theorem~\ref{thm:auto-expansion}. The smallest $\Sigma$-invariant set containing $\rows{H}$ is the right orbit union
    \begin{equation}
        \mathcal O = \bigcup_{r\in \rows{H}}\orbs{r} = \{r\cdot\sigma\ |\ r\in \rows{H}, \sigma\in\Sigma\},
    \end{equation}
    which gives $H_\mathcal{O}$ constructed from stacking rows in $\mathcal{O}$. For each $r\in \mathcal{O}$, we have $r =\orbs{r'}$ for some $r'\in \rows{H}$, hence for any $\sigma\in\Sigma$, $r\cdot\sigma\in\orbs{r'}\in\mathcal{O} = \rows{H_\mathcal{O}}$, and it follows $H_\mathcal{O}\sigma = \rho H_\mathcal{O}$ for all $\sigma\in \Sigma$.
\end{proof}

The number of checks can be improved if the automorphism group was obtained through code expansion, and LDPC properties can still be maintained if both the size and the sparsity of prescribed automorphism is bounded :

\begin{theorem}[Expanded Matrix Automorphism Conversion]
Let us be given $\mathcal{C}(G,H)$ and $\mathcal{G} \leq \gln{k}$. Suppose $H$ is $w$-bounded, $g$ is $t$-bounded for all $g\in \mathcal{G}$ and $|\mathcal{G}| = m$. Then, there exists a family of $[[nm, k, dm]]$ codes $\mathcal{C}_\mathcal{G}(G_\mathcal{G},H_\mathcal{G})$ with a $(w +mt)$-bounded check matrix $H_\mathcal{G}$ using $(m-1)mk$ additional checks such that  $g\in \mathcal{G}$ is a matrix automorphism for all $g\in \mathcal{G}$.
\end{theorem}
\begin{proof}
    Let $G_\mathcal{G} = \begin{bmatrix}I&g_1\cdots g_m&M&g_1M\cdots g_mM\end{bmatrix}$ as before. Left action by $g\in\mathcal{G}$ induces a natural column permutation $\sigma_c = I_{(m+1)k}\otimes \sigma_G\oplus I_{(m+1)(n-k)}\otimes \sigma_G$ where $\sigma_g(i) = j$ if $gg_i = g_j$. Furthermore, we note that both check matrices given in Theorem~\ref{thm:expansion-ldpc} contain $H_{lower} = I_{(m+1)k}\otimes H_1\oplus I_{(m+1)(n-k)}\otimes P_H$ which commutes with $\sigma_c$. It suffice to augment checks that ensure the upper half is also preserved, yielding the $(m-1)mk$ additional checks and $(m-1)t$ contribution to column weight.
\end{proof}

Similar to before, when we have more information about $\Sigma$ (or the related $\mathcal{G}$), we can come up with sharper bounds:

\begin{corollary}[Conversion: Local CNOTs]
 Let us be given a code $\mathcal C(G,H)$ and let $(g, \sigma)\in \mathcal{A}(G)$ be an order-2 CNOT circuit on $t$ bits where $gG\sigma = G$. There exists a valid check matrix $H'$ using $t$ additional checks, and a row permutation $\rho\in P_{n+t}$, such that $\rho H'\sigma = H'$. Furthermore, if $H$ is $w$ bounded, $H'$ is is $w +t+1$-bounded.
\end{corollary}

\section{CQLU Logical Operations}\label{appx:cqlu}
Given the matrix automorphism guarantees, as shown in \cite{berthusen2025automorphismgadgetshomologicalproduct}, we can construct Hypergraph product and Homological product codes that inherits the code automorphisms as fault tolerant logical operations implementable via qubit relabeling. In this section, we focus on HGPS codes and study their matrix automorphism groups. 

\subsection{Dirty Shifts}

First, let us construct the basis needed to realize the dirty cyclic shifts as matrix automorphisms.

\begin{proposition}\label{prop:dirty-shift}
    Let $P_i$ be the cyclic shift matrix on $i$ elements. Given the binary simplex code $\mathcal{S}_r$ with a circulant check matrix $H_r$. Then, $P_{2^{r}-1}$ is a matrix automorphism, and exist a basis $G$ for which the corresponding logical action of $P_{2^{r}-1}$ is
    \begin{equation}
        g_{ds} = \begin{bmatrix}
            &&&&&*\\
            1&&&&&*\\
            &1&&&&*\\
            &&\ddots&&&\vdots\\\\
            &&&&1&*
        \end{bmatrix},
    \end{equation}
    which implements a cyclic row permutation of $G$ plus a CNOT-fanout gate. 
\end{proposition}

\begin{proof}
    It is easy to see that the column shift matrix $P = P_{2^r-1}$ is indeed a matrix automorphism by definition of it being circulant. Suppose the logical action is $g$, we can also show that $g$ has order $2^{r}-1$, following from Fact~\ref{fact:auto} and since $P$ has order $2^{r}-1$. 

    Now, let $c\in \mathcal{S}_r$ be a non-zero codeword and define $B = \{c_i|c_i = c\cdot P^i, 0\leq i<r\}$. Suppose $B$ is a spanning set (for example, picking $c = [1,1,\cdots1,0,0,\cdots0]$ works). Then, the generator matrix $G_B$ built from stacking $v\in B$ as rows is the desired basis. Indeed, for each $i<r-2$, $c_iP = c_{i+1}$; column permutation by $P$ shifts rows of $G_B$ by one. Furthermore, since $\braket{B}$ is $r$-dimensional and $P$ is a code automorphism, $\braket{B} = \mathcal{S}_r$, and that $c_{r-1}P\in \mathcal{S}_r = \sum_{0\leq i<r}b_i c_i$. Let $g$ be the operator where $gB = BP$; it follows that 
    \begin{equation}
        g_{ds} = \begin{bmatrix}
            &&&&b_0\\
            1&&&&b_1\\
            &\ddots&&&\vdots\\\\
            &&&1&b_{r-1}
        \end{bmatrix},
    \end{equation}
    as desired.
    
\end{proof}

\subsection{autoCNOTs}
For $\mathcal{S}_4$, there exists non-trivial matrix automorphisms as well using $H_4$ given in Table~\ref{tab:simplex} that can implement a order-2 logical circuit given a basis. Furthermore, by choosing different generator matrices, we can obtain a wide range of codes with various logical operators in the same conjugacy class. 
\begin{proposition}
    Let $\tau_4$ be the permutation matrix on $16$ elements that performs a transposition on $4\times 4$ square of the $16$ elements in lexicographical order. That is, suppose $i = 4a+b$ for some $a,b<4$, then $\tau(i = 4a+b)\mapsto  j = 4b+a$. Let $P_{auto}$ be the permutation obtained by removing the first element from $\tau_4$. Then,  $P_{auto}$ is a matrix automorphism for $\mathcal{S}_4$ with $H_4 = I +P_{15}+P_{15}^4$. Furthermore, there exist a basis $G$ for which the corresponding logical action of $P_{auto}$ is $g$ iff $g$ is conjugate to  $g_{SWAP} = I\otimes X$; that is, if there exists $h$ such that:
    \begin{equation}
        hgh^{-1} = \begin{bmatrix}
            &1&&\\
            1&&&\\
            &&&1\\
            &&1&
        \end{bmatrix}.
    \end{equation}
\end{proposition}
Note, by choosing $h' = \begin{bmatrix}
        1&&&\\
            &1&&\\
            1&&1&\\
            &1&&1    \end{bmatrix}$, we obtain $g_{auto} = \begin{bmatrix}
        1&&&\\
            1&1&&\\
            1&&1&\\
            1&1&1&1
    \end{bmatrix}$ used in the main text. 

\begin{proof}
Note $P_{auto}$ is an involution, and $P_{auto} P_{15}P_{auto} = P_{15}^4$: let each $i = 4a+b$ be represented by an ordered tuple $(a,b)$, we have $P_{auto} P_{15}P_{auto}(a,b) =P_{auto} P_{15}(b,a) = P_{auto}(b,a+1) = (a+1,b) = P^4_{15}(a,b)$. It follows that $H_{4}P_{auto} = H_4$. 

To give a basis where the resulting logical action is a disjoint pair of SWAPs, we can find $c_1, c_2$ such that the basis $\{c_1, c_1P_{auto}, c_2,c_2P_{auto}\}$ is spanning. Let this generator matrix be $G_{SWAP}$. Then, for any invertible $h$, $G' = hG$ is also a valid generator matrix, and applying $P_{auto}$ on $G'$ gives $g$:
\begin{align*}
    G'P_{auto} &= hG_{SWAP}P_{auto} = hg_{SWAP}G_{SWAP} =gG'\\&\iff g = hg_{SWAP}h^{-1}
\end{align*}

which is conjugate to $g_{SWAP}$, as desired.
\end{proof}

There are many choices that satisfy both propositions, giving us a code with desired ISAs.

While the primitive polynomial for $\mathcal{S}_4$ allows for a wide class of logical actions to be embedded, it is not guaranteed to be the case for larger $r$. Indeed, we find that the primitive polynomials as reported in Table~\ref{tab:simplex} do no admit similar symmetries for any $r>4$ via a numerical search, which means code modifications are unavoidable. However, beyond the check extensions already presented, there may exist circulant matrices with more terms that enforces the symmetry, but we leave this as future work to perform more intensive search.  
\subsection{Global Transversal Operations}
In addition to automorphism gates, we briefly describe other transversal gates we can apply to the HGPS code. First, by virtue of the HGP construction, we can apply puncturing/augmenting techniques to $\mathcal{S}_4$ to achieve homomorphic CNOTs and homomorphic measurements between different code blocks, as done in \cite{hgphomo}. By virtue of being a square symmetric code, the HGPS codes can access transversal $H$-$SWAP$ gate, $CZ$-$S$ gate, and sibiling-$CZ$s as give in \cite{partition}. These gates allow us to implement:
\begin{enumerate}
    \item Hadamard on all logical qubits, plus a logical transpose,
    \item $S$/$S^\dagger$ gates on $2\sqrt{k}$ diagonal qubits, and $CZ$s on transpose-correspondent logical qubits in the same sector,
    \item $CZ$ gates between correspondent logical qubits in both sectors.
\end{enumerate}
These operations allow us to produce and consume $\ket{i}$ state on the diagonal efficiently. For a more complete treatement, we refer readers to \cite{hgphomo, partition}.

\section{Circuit Compilations}\label{appx:circuits}
In this section, we give the exact CQLU circuit implementation for quantum adder and related state preparation.

Figure~\ref{fig:adder-full} shows the full adder derivation:

\begin{itemize}
    \item \emph{Step 1}: we propagate the Pauli corrections arising from Bell-measurement through the circuit, possibly flipping the $S$-correction from $\ket{T}$ injections. 
    \item \emph{Step 2}: commuting the final CNOT in the MAJ block towards past the measurement bell. This creates an additional CNOT correction highlighted in blue. 
    \item \emph{Step 3-5}: commuting the $S$ corrections past bell measurements. 
    \item \emph{Step 6a}: commuting the first CNOT in the MAJ (highlighted in green) past the fan-in, adding an additional CNOT correction (dashed-grey). 
    \item \emph{Step 6b} The newly added CNOT is then commuted past the bell-pair generation, leading to two CNOTs highlighted in blue, one of which acts trivially. \
    \item emph{Step 7a}: adding two pairs of CNOTs and commuting them past the fanouts and $T$ gates in the middle. 
    \item \emph{Step 7b}: the additional CNOT corrections cancel each other out. 
    \item \emph{Step 8-9}: finally, we conjugate by global $H$ to change orientations of the CNOTs, and use a transversal SWAP to prepare the block for reactive measurements.   
\end{itemize}

Figure~\ref{fig:compiled-adder} gives a spacial layout and exact CQLU operations required, alone with the state-prep step that can be done offline. This circuit has at most $4$ active blocks in any time step. 

Derivations for UMA and temporary AND-Toffoli (used in QROM) can be obtained following similar transformations, which we do not describe explicitly here. The QROM circuit used in \cite{GidneyLinearT} consist of temperoray-AND Toffolis and CNOT fanouts, which can be performed by consuming a GHZ state on $n$-qubits \cite{HarnessingGHZ}, which is included within Figure~\ref{fig:compiled-adder}.
\begin{figure}
    \centering
    \includegraphics[width=\linewidth]{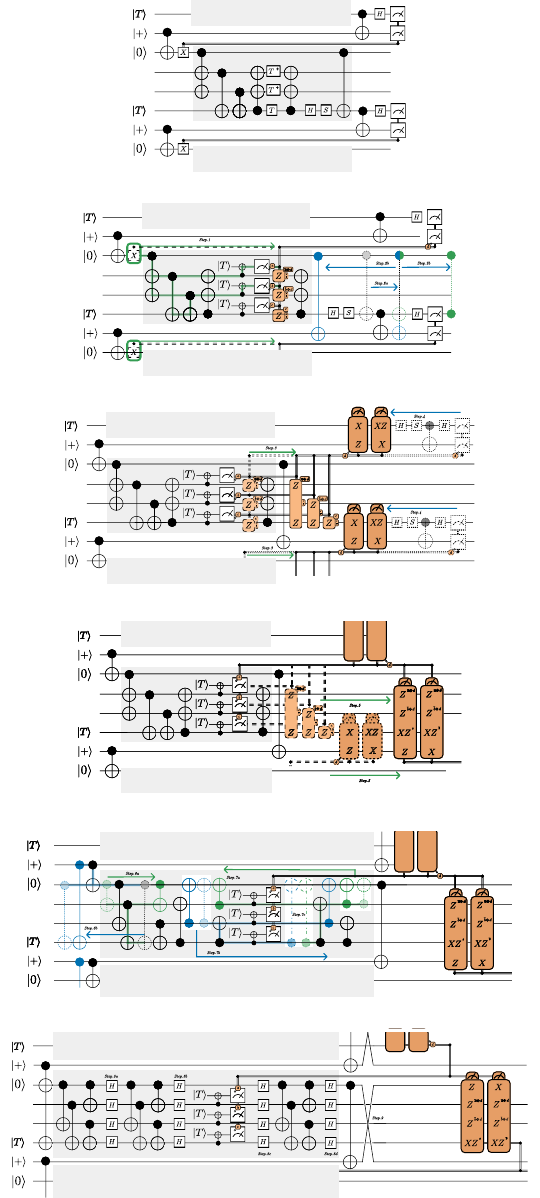}
    \caption{Full derivation of the MAJ blocks implemented using CQLU instructions. Pauli corrections that commutes with non-Clifford operations are tracked in software. 
    }
    \label{fig:adder-full}
\end{figure}
\begin{figure*}
    \centering
    \includegraphics[width=0.95\linewidth]{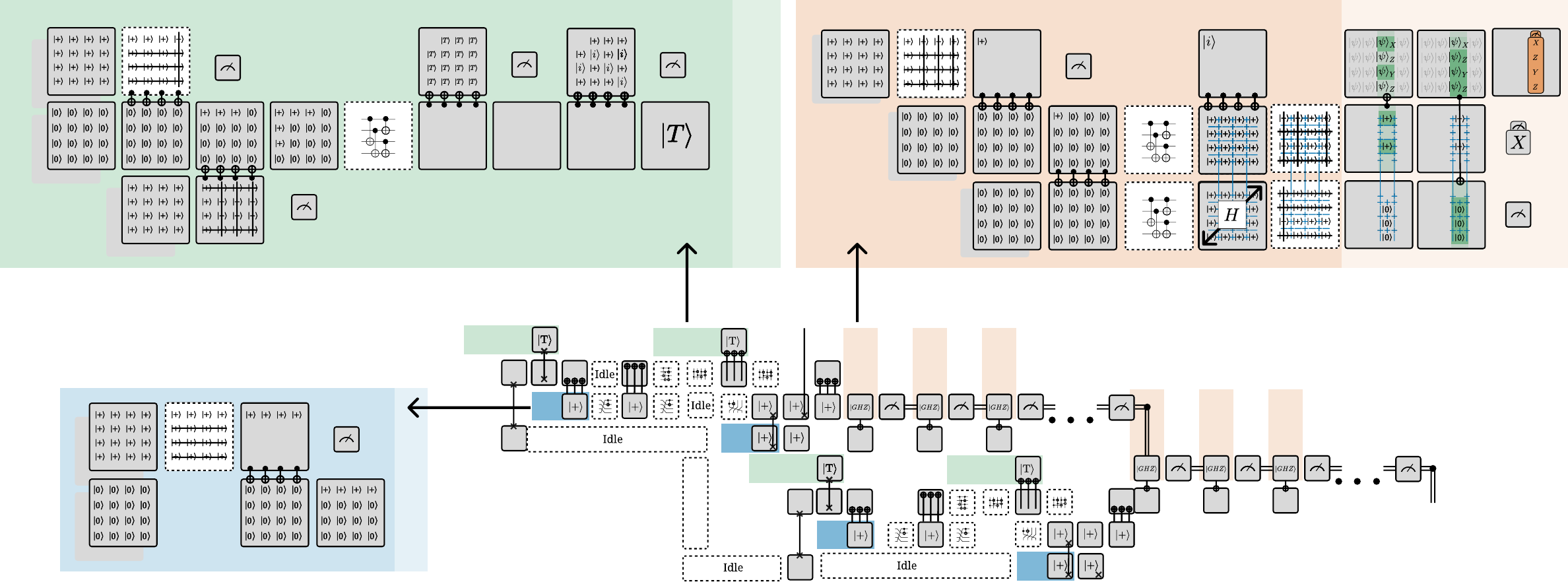}
    \caption{The exact compiled adder circuit, with offline state factories produced. Initial state preparation (blue) and $T$ preparation (green) can be done completely offline. GHZ states for reactive measurements (orange) can be carried out mostly offline (dark orange), with two additional CNOTs  and a measurement (light orange). }
    \label{fig:compiled-adder}
\end{figure*}
\newpage
% \bibliography{refs} 

\end{document}